\documentclass[journal]{IEEEtran}
\usepackage[cmex10]{amsmath}
\usepackage{amssymb}

\usepackage[dvips]{graphicx}

\usepackage{cite}

\graphicspath{{./}}

\usepackage{setspace}

\usepackage{array}
\usepackage{fixltx2e}
\usepackage{times,color}

\renewcommand{\P}{\mathbb{P}}
\newcommand{\mc}{\mathcal}
\newcommand{\mb}{\mathbb}
\newcommand{\ms}{\mathsf}
\newcommand{\E}{\mathbb{E}}

\newtheorem{theorem}{Theorem}
\newtheorem{lemma}{Lemma}
\newtheorem{definition}{Definition}

\newtheorem{corollary}{Corollary}

\begin{document}

\title{Feedback Communication Systems with Limitations on Incremental Redundancy}
\author{Tsung-Yi~Chen,~\IEEEmembership{Student~Member,~IEEE,}
Adam~R.~Williamson,~\IEEEmembership{Student~Member,~IEEE,}
Nambi~Seshadri,~\IEEEmembership{Fellow,~IEEE,}
and~Richard~D.~Wesel,~\IEEEmembership{Senior~Member,~IEEE}
\thanks{Richard~D.~Wesel, Adam~R.~Williamson and Tsung-Yi Chen are with the Department of Electrical Engineering, University of California, Los Angeles, Los Angeles, CA, 90095, USA.}
\thanks{Nambi Seshadri is with Broadcom Corporation, Irvine, CA, 92617, USA.}%
\thanks{This research was supported by a gift from the Broadcom Foundation.  Dr. Wesel has also consulted for the Broadcom Corporation on matters unrelated to this research.}%
\thanks{This research was supported by National Science Foundation Grant CIF CCF 1162501.}
}

\markboth{}{Submitted paper}

\maketitle


\begin{abstract}

This paper explores feedback systems using incremental redundancy (IR) with noiseless transmitter confirmation (NTC). For IR-NTC systems based on {\em finite-length} codes (with blocklength $N$) and decoding attempts only at {\em certain specified decoding times}, this paper presents the asymptotic  expansion achieved by random coding, provides rate-compatible sphere-packing (RCSP) performance approximations, and presents simulation results of tail-biting convolutional codes.

The information-theoretic analysis shows that values of $N$ relatively close to the expected latency yield the same random-coding achievability expansion as with $N = \infty$.  However, the penalty introduced in the expansion by limiting decoding times is linear in the interval between decoding times. For binary symmetric channels, the RCSP approximation provides an efficiently-computed approximation of performance that shows excellent agreement with a family of rate-compatible, tail-biting convolutional codes in the short-latency regime.  For the additive white Gaussian noise channel, bounded-distance decoding simplifies the computation of the marginal RCSP approximation and produces similar results as analysis based on maximum-likelihood decoding for latencies greater than $200$. The efficiency of the marginal RCSP approximation facilitates optimization of the lengths of incremental transmissions when the number of incremental transmissions is constrained to be small or the length of the incremental transmissions is constrained to be uniform after the first transmission.  Finally, an RCSP-based decoding error trajectory is introduced that provides target error rates for the design of rate-compatible code families for use in feedback communication systems.

\end{abstract}


\section{Introduction}
This section presents a summary of related previous work on noiseless feedback in communication from both the perspective of information theory and that of system design.  After discussing previous work, we summarize this paper's main contributions.

\subsection{Feedback in Information Theory}
The first notable appearances of feedback in information theory are in 1956. Shannon showed the surprising result that even in the presence of full feedback (i.e., instantaneous feedback of the received symbol) the capacity for a single-user memoryless channel remains the same as without feedback  \cite{Shannon_zero_1956}. That same year Elias and Chang each showed that while feedback does not change the capacity, it can make approaching the capacity much easier.  Elias \cite{Elias_channel_1956}  showed that feedback makes it possible to achieve the best possible distortion using a simple strategy without any coding when transmitting a Gaussian source over a Gaussian channel if the channel bandwidth is an integer multiple of the source bandwidth.  Chang \cite{Chang_theory_1956} showed that feedback can be used to reduce the equivocation $H(X|Y)$ in a coded system to bring its rate closer to capacity for finite blocklengths.
   
The next phase in the information-theoretic analysis of feedback showed that it can greatly improve the error exponent.  Numerous papers including \cite{Schalkwijk_1966_1, Schalkwijk_1966_2, Kramer_1969, Zig_1970, Nakiboglu_2008, Gallager_2010} explicitly showed this improvement.  A key result in this area is the seminal work of Burnashev \cite{Burnashev_1976} which  provides an elegant expression of the optimal error-exponent for DMC with noiseless feedback.  Burnashev employed a technique that can be considered as a form of active hypothesis testing \cite{Naghshvar_active_2012}, which uses feedback to adapt future transmitted symbols based on the current state of the receiver.   

Although active hypothesis testing provides a performance benefit,  practical systems using feedback have thus far primarily used feedback to explicitly decide when to stop transmitting additional information about the intended message.  This can happen in two ways:  Receiver confirmation (RC) occurs when the receiver decides whether it has decoded with sufficient reliability (e.g. passing a checksum) to terminate communication and feeds this decision back to the transmitter.  The alternative to RC is transmitter confirmation (TC), in which the transmitter decides (based on feedback from the receiver) whether the receiver has decoded with sufficient reliability (or even if it has decoded correctly, since the transmitter knows the true message).  

TC schemes often use distinct transmissions for a message phase and a confirmation phase. Practical TC and RC systems can usually be assigned to one of two categories based on when confirmation is possible.   Single-codeword repetition (SCR) only allows confirmation at the end of a complete codeword and repeats the same codeword until confirmation.  In contrast, incremental redundancy (IR) systems \cite{Mandelbaum_adaptive_1974}  transmit a sequence of distinct coded symbols with numerous opportunities for confirmation within the sequence before the codeword is repeated.  In some cases of IR, the sequence of distinct coded symbols is infinite and is therefore never repeated.  If the sequence of symbols is finite and thus repeated, we call this a repeated IR system.

Forney's analysis  \cite{Forney_exponential_1968} provided an early connection between these practical system designs  and theoretical analysis by deriving error-exponent bounds for a DMC using an SCR-RC scheme. Following Forney's work, Yamamoto and Itoh \cite{Yamamoto_1979} replaced Forney's SCR-RC scheme with a SCR-TC scheme in which the receiver feeds back its decoding result (based only on the codeword sent during the current message phase). The transmitter confirms or rejects the decoded message during a confirmation phase, continuing with additional message and confirmation phases if needed.  This relatively simple SCR-TC scheme allows block codes to achieve the optimal error-exponent of Burnashev for DMCs. 


Error-exponent results are asymptotic and do not always provide the correct guidance in the low-latency regime.  Polyanskiy et al. \cite{PolyIT11} analyzed the benefit of feedback in the non-asymptotic regime. They studied random-coding IR schemes under both RC and TC, and showed that capacity can be closely approached in hundreds of symbols using feedback \cite{PolyIT11} rather than the thousands of symbols required without feedback \cite{PolyIT10}.   

In \cite{PolyIT11} the IR schemes provide information to the receiver one symbol at a time.  For the RC approach, the receiver provides confirmation when the belief for a prospective codeword exceeds a threshold.  For the TC approach, the transmitter provides confirmation using a special noiseless transmitter confirmation (NTC) allowed once per message when the transmitter observes that the currently decoded message is correct. This setup, referred to as variable-length feedback codes with termination (VLFT) in \cite{PolyIT11}, supports zero-error communication. VLFT codes are defined broadly enough to include active hypothesis testing in \cite{PolyIT11}, but the achievability results of interest use feedback only for confirmation.

When expected latency (average blocklength) is constrained to be short, there is a considerable performance gap shown in \cite{PolyIT11} between the SCR-TC scheme of Yamamoto and Itoh \cite{Yamamoto_1979} and the superior IR-TC scheme of \cite{PolyIT11}.  This is notable because the SCR-TC scheme of Yamamoto and Itoh achieves the best possible error-exponent, demonstrating that error-exponent results do not always provide the correct guidance in the non-asymptotic regime.

Using \cite{PolyIT11} as our theoretical starting point, this paper explores the information-theoretic analysis of IR-NTC systems at low latencies (i.e., with short average blocklengths) under practical constraints and with practical codes.

\subsection{Feedback in Practical Systems}

The use of feedback in practical communication systems predates Shannon's 1948 paper.  By 1940 a patent had been filed for an RC feedback communication system using ``repeat request'' in printing telegraph systems \cite{ARQpatent}. The first analysis of a retransmission system using feedback in the practical literature appears to be \cite{Benice_An_1964} in 1964, which  analyzed automatic repeat request (ARQ) for uncoded systems that employ error detection (ED) to determine when to request a repeat transmission. 

In the 1960s, forward-error correction (FEC) and ED-based ARQ were considered as two separate approaches to enhance transmission reliability. Davida proposed the idea of combining ARQ and FEC \cite{Davida_Forward_1972} using punctured, linear, block codes.  The possible combinations of FEC and ARQ came to be known as type-I and type-II hybrid ARQ (HARQ), which first appeared in \cite{Lin_hybrid_1982} (also see \cite{BookLin}). Type-I HARQ is an SCR-RC feedback system that repeats the same set of coded symbols with both ED and FEC until the decoded message passes the ED test.  Recent literature and the present paper refer to type-I HARQ as simply ARQ because today's systems rarely send uncoded messages.  Type-II HARQ originally referred to systems that alternated between uncoded data with error detection and a separate transmission of parity symbols.  Today, type-II HARQ has taken on a wider meaning including essentially all IR-RC feedback systems.  

Utilizing the soft information of the received coded symbols, Chase proposed a decoding scheme for SCR-RC \cite{Chase_1985} that applies maximal ratio combining to the repeated blocks of coded symbols. Hagenauer \cite{Hagenauer_Rate_1988} introduced rate-compatible punctured convolutional (RCPC) codes, which allow a wide range of rate flexibility for IR schemes. Rowitch et al. \cite{Rowitch_performance_2000} and Liu et al.\cite{Liu_punctured_2003} used rate-compatible punctured turbo (RCPT) codes in an IR-RC system to match the expected throughput to the binary-input AWGN channel capacity.    

When implementing an IR system using a family of rate-compatible codes, the IR transmissions will often repeat once all of the symbols corresponding to the longest codeword have been transmitted and the confirmation is not yet received. Chen et al. demonstrated in \cite{Chen_2010_ITA} that a repeated IR system using RCPC codes could deliver bit error rate performance similar to long-blocklength turbo and LDPC codes, but with much lower latency.  The demonstration of \cite{Chen_2010_ITA} qualitatively agrees with the error-exponent analysis and the random-coding analysis  in \cite{PolyIT11}.  Using \cite{Chen_2010_ITA} as our practical starting point, this paper explores the practical design of IR-NTC systems at low latencies (i.e., with short average blocklengths) under practical constraints and with practical codes. 

Practical design for an IR scheme is primarily concerned with optimizing the incremental transmissions, which includes design of a rate-compatible code family and optimization of the lengths of the incremental transmissions.  Practical constraints include limitations on the blocklength $N$ of the lowest-rate codeword in the rate-compatible code family, the number $m$ of incremental transmissions, and the lengths of the incremental transmissions. 

Several papers have provided inspiration in practical deign of feedback systems.  Uhlemann et al. \cite{Uhlemann_Optimal_IR_ISIT_2003} studied a similar optimization problem assuming that the error probabilities of each transmission are given. Finding code-independent estimates of appropriate target values for these error probabilities is one goal of this paper.


In related work, \cite{Visotsky_RBIR_TCOM_2005} provides expressions relating throughput and latency to facilitate optimization of the lengths of incremental transmissions in real-time for an IR-RC system. In contrast with the current paper, the receivers  in \cite{Visotsky_RBIR_TCOM_2005} must send the optimized transmission lengths in addition to the confirmation message in order to adapt the transmission lengths in real time. Moreover, \cite{Visotsky_RBIR_TCOM_2005} used relatively large transmission lengths in order to accurately approximate mappings from codeword reliability to block error rate, while our focus is on short blocklengths. 

Chande et al. \cite{Chande_ARQ_DP_ISIT_1998} and Visotsky et al. \cite{Visotsky_ARQ_DP_ISIT_2003} discussed how to choose the optimal blocklengths for incremental redundancy using a dynamic programming framework.   Their approach uses the error probabilities of a given code with specified puncturing to determine the optimal transmission lengths for that code.  That work focused on fading channels and on longer blocklengths than the current paper.  

Freudenberger's work \cite{Freudenberger_Unreliable_TCOM_2004} differs from the present work in that the receiver requests specific incremental retransmissions by determining which segments of the received word are unreliable.  This represents a step in the direction of a practical active hypothesis testing system.

In \cite{Fricke_Reliability_HARQ_TCOM_2009} the authors used a reliability-based retransmission criteria in a HARQ setting to show the throughput gains compared to using cyclic redundancy checks for error detection.  While \cite{Fricke_Reliability_HARQ_TCOM_2009} shows simulation results for several values of the message size $k$, the results are not discussed in terms of proximity to capacity at short blocklengths.

\subsection{Main Contributions}
\label{sec:MainContrib} 
This paper is concerned with IR-NTC feedback systems that use a rate-compatible code family to send $m$ incremental transmissions. The $j$th incremental transmission has length $I_j$ so that after $i$ incremental transmissions the cumulative blocklength is $n_i=\sum_{j=1}^i I_j$ and the decoder uses all $n_i$ symbols to decode.  We are also concerned with the blocklength of the lowest-rate codeword in the rate-compatible code family, which is $N=\sum_{j=1}^m I_j$.


Among its numerous other results, \cite{PolyIT11} provides achievability result for an IR-NTC system with zero-error probability in the non-asymptotic regime. The achievability uses random-coding analysis for the case where $N=\infty$ (and thus $m=\infty$) and $I_j=1 ,\forall j$. 

The main contributions of this paper are as follows:
\begin{enumerate}
\item This paper extends the IR-NTC random-coding result of \cite{PolyIT11}  to include the constraint of finite $N$  and of $I >1$ where $I_j = I ~ \forall j$. Though $N$ is finite, our setup still supports zero-error probability through an ARQ-style retransmission if decoding fails even after the $N$ symbols of the lowest-rate codeword have been transmitted. We refer this type of system as repeated IR-NTC.  Our analysis yields two primary conclusions:
\begin{itemize}
\item Constraining the difference between $N$ and the expected latency to grow logarithmically with the expected latency negligibly reduces expected throughput as compared with $N=\infty$.
\item An uniform incremental transmission length $I$ causes expected latency to increase linearly with $I$ as compared to the $I=1$ case.
\end{itemize}
\item This paper uses rate-compatible sphere-packing (RCSP) approximation as a tool for analysis and optimization of IR feedback systems as follows:
\begin{itemize}
\item Exact joint error probability computations based on RCSP and bounded-distance (BD) decoding are used to optimize the incremental lengths $I_1, \ldots I_m$ for small values of $m$ in a repeated IR-NTC system.
\item Tight bounds on the joint RCSP performance under BD decoding are used to optimize a fixed incremental length $I$ ($I_j = I ~ \forall j$) for repeated IR-NTC with larger values of $m$.
\item  For non-repeating IR-NTC, exact joint RCSP computations under BD decoding are used to optimize the incremental lengths $I_1, \ldots I_m$ to minimize expected latency under an outage constraint. 
\item RCSP approximations yield decoding error trajectory curves that provide design targets for families of rate-compatible codes for IR-NTC feedback systems.
\end{itemize}
\item This paper provides simulations of IR-NTC systems  based on randomly-punctured tail-biting convolutional codes that demonstrate the following:
\begin{itemize} 
\item For $I=1$ and large $N$ these simulations exceed the random-coding lower bound of \cite{PolyIT11} on both the BSC and the AWGN channel.  For the BSC, performance closely matches the RCSP approximation at low latencies.  For the AWGN channel, performance is also similar, but with a gap possibly due to the constraint of the convolutional codes to a binary-input alphabet.
\item Simulations were also performed on the AWGN channel using $m=5$ incremental transmissions with transmission increments $I_j$ optimized using RCSP under the assumption of BD decoding.  These simulations exceeded the random-coding lower bound of \cite{PolyIT11} at low latencies for the same values of $m$ and $I_j$.  Also at sufficiently low latencies, the simulations closely match the throughput-latency performance predicted by RCSP with ML decoding. 
\end{itemize}
\end{enumerate}

\subsection{Organization}

To conclude this section, we briefly summarize the organization of this paper. Sec.~\ref{sec:VLFT} investigates the rate penalty incurred by imposing the constraints of finite blocklength and limited decoding times on VLFT codes.  A numerical example for the binary symmetric channel (BSC) is presented in Sec.~\ref{sec:NumResults}.
Sec.~\ref{sec:VLFTexamples} describes the examples (used throughout the rest of the paper) of IR-NTC systems using families of practical rate-compatible codes based on convolutional and turbo codes. Sec.~\ref{sec:IR_RCU_RCSP} presents the RCSP approximations and applies the approximations to provide numerical examples for the BSC and the additive white Gaussian noise (AWGN) channel. Sec.~\ref{sec:IR_RCU_RCSP} also explores ARQ, Chase code combining, and IR-NTC systems for the AWGN channels using RCSP approximation.  Sec.~\ref{sec:Optimization} continues the RCSP-based exploration of the use of IR-NTC on the AWGN channels studying the optimization of the increments $\{I_j\}_{j = 1}^m$ based on RCSP approximation. Finally Sec.~\ref{sec:Conclusion} concludes the paper.

%
\section{Practical Constraints on VLFT}

\label{sec:VLFT}

This section studies IR-NTC systems using the VLFT framework in \cite{PolyIT11} with the practical constraints of finite $N$ and uniform increments $I > 1$.

\subsection{Review of VLFT Achievability}
\label{sec:Prelim}
We will consider discrete memoryless channels (DMC) in this section and use the following notation: $x^n = (x_1, x_2, \dots, x_n)$ denotes an $n$-dimensional vector, $x_j$ the $j$th element of $x^n$, and $x_i^j$ the $i$th to $j$th elements of $x^n$. We denote a random variable (r.v.) by capitalized letter, e.g., $X^n$, unless otherwise stated. The input and output alphabets are denoted as $\mc{X}$ and $\mc{Y}$ respectively.  Let the input and output product spaces be $\ms{X} = \mc{X}^n, \ms{Y} = \mc{Y}^n$ respectively. 
A channel used without feedback is characterized by a conditional distribution $P_{\ms{Y}|\ms{X}} = \prod_{i = 1}^n P_{Y_i|X_i}$ where the equality holds because the channel is memoryless. For codes that use a noiseless feedback link, we consider causal channels $\{P_{Y_i|X_1^iY^{i-1}}\}_{i = 1}^{\infty}$ and additionally focus on causal, memoryless channels $\{P_{Y_i|X_i}\}_{i = 1}^{\infty}$. 

Let the finite dimensional distribution of $(X^n, \bar{X}^n, Y^n)$ be: 
\begin{align}
\begin{split}
&P_{X^nY^n\bar{X}^n}(x^n,y^n,\bar{x}^n)
\\
&= P_{X^n}(x^n)P_{X^n}(\bar{x}^n)\prod_{j = 1}^nP_{Y_j|X^jY^{j-1}}(y_j|x^j, y^{j-1}) ,
\end{split}
\end{align}
i.e., the distribution of $\bar{X}^n$ is identical to $X^n$ but independent of $Y^n$. The information density $i(x^n;y^n)$ is defined as
\begin{align}
\label{eqn:InformationDensity}
i(x^n; y^n) &= \log \frac{dP_{X^nY^n}(x^n,y^n)}{d(P_{X^n}(x^n)\times P_{Y^n}(y^n))} 
\\
&= \log \frac{dP_{Y^n|X^n}(y^n|x^n)}{dP_{Y^n}(y^n)}.
\end{align}
In this paper we only consider channels with essentially bounded information density $i(X;Y)$.

This section extends the results in  \cite{PolyIT11} for VLFT codes.  In order to be self-contained, we state the definition of  VLFT codes in \cite{PolyIT11}:
\begin{definition}
An $(\ell,M,\epsilon)$ variable-length feedback code with termination (VLFT code) is defined as:
	\begin{enumerate} 
	\item A common r.v. $U \in \mc{U}$ with a probability distribution $P_U$ revealed to both transmitter and receiver before the start of transmission.
	\item A sequence of encoders $f_n:\mc{U}\times\mc{W}\times \mc{Y}^{n-1} \rightarrow \mc{X}$ that defines the channel inputs $X_n = f_n(U,W,Y^{n-1})$.  Here $W$ is the message r.v. uniform in $\mc{W} = \{1, \dots, M\}$.
	\item A sequence of decoders $g_n: \mc{U}\times\mc{Y}^n \rightarrow \mc{W}$ providing the estimate of $W$ at time $n$.
	\item A stopping time $\tau \in \mb{N}$ w.r.t. the filtration $\mc{F}_n = \sigma\{U, Y^n, W\}$ such that: 
	\begin{align}
	\E[\tau]\leq \ell.
	\end{align}
	\item The final decision $\hat{W} = g_\tau(U, Y^\tau)$ must satisfy:
	\begin{align}
	\P[\hat{W} \ne W]\leq \epsilon.
	\end{align}
	\end{enumerate}
\end{definition}

VLFT represents a TC feedback system because the stopping time defined in item 4) above has access to the message $W$, which is only available at the transmitter.   As observed in \cite{PolyIT11}, the setup of VLFT is equivalent to augmenting each channel with a special use-once input symbol, referred to as the termination symbol, that has infinite reliability. We will refer this concept as noiseless transmitter confirmation (NTC) for the rest of the paper. 

NTC simplifies analysis by separating the confirmation/termination operation from regular physical channel communication. The assumption of NTC captures the fact that many practical systems communicate control signals in upper protocol layers. 

The assumption of NTC increases the non-asymptotic, achievable rate to be larger than the original feedback channel capacity because it noiselessly provides the information of the stopping position.  This increases the capacity by the conditional entropy of the stopping position given the received symbols normalized by the average blocklength.

The following is the zero-error VLFT achievability in \cite{PolyIT11}:
\begin{theorem}[\cite{PolyIT11}, Thm. 10]
\label{thm:RC_Achieve}
Fixing $M > 0$, there exists an $(\ell, M, 0)$ VLFT code with
\begin{align}
\label{eqn:AchevVLFT}
\ell &\leq \sum_{n = 0}^{\infty} \xi_n
\end{align}
where $\xi_n$ is the following expectation: 
\begin{equation} \label{eqn:Xi_n}
\small
\xi_n=\E \min\left\{1, (\text{$M$$-$$1$})\P[i(X^n;Y^n)\leq i(\bar{X}^n;Y^n)|X^nY^n]\right\} \hspace{-.05in}.
\end{equation}
The expression in \eqref{eqn:Xi_n} is referred to as the random-coding union (RCU) bound.
We take the information density before any symbols are received,  $i(X^0;Y^0)$, to be $0$ and hence $\xi_0 = 1$.  Additionally, from the proof of \cite[Thm. 11]{PolyIT11}, we have: 
\begin {equation}
\label{eqn:VLFT_DT}
\xi_n \leq \E\left[ \exp\left\{-[i(X^n; Y^n) - \log (M-1)]^+\right \}\right]. 
\end{equation}
\end{theorem}

The proof of Thm.~\ref{thm:RC_Achieve} is based on a special class of VLFT codes called fixed-to-variable (FV) codes \cite{Verdu_10}. FV codes satisfy the following conditions:
\begin{align}
\label{eqn:StopFeedback}
&f_n(U,W,Y^{n-1}) = f_n(U,W)
\\
\label{eqn:StopZeroError}
&\tau = \inf\{ n \geq 1: g_n(U,Y^n) = W\}.
\end{align}
The condition in \eqref{eqn:StopFeedback} precludes active hypothesis testing since the feedback is not used by the encoder to determine transmitted symbols.  The condition in \eqref{eqn:StopZeroError} enforces zero-error operation since the stopping criterion is correct decoding.

In the proof of Thm.~\ref{thm:RC_Achieve} each codeword is randomly drawn according to the capacity-achieving input distribution $\prod_{j = 1}^{\infty}P_{X}$ on the infinite product space $X^{\infty}$. Given a codebook realization, the encoder maps a message to an infinite-length vector and the transmitter sends the vector over the channel symbol-by-symbol. Upon receiving each symbol, the decoder computes $M$ different information densities between the $M$ different codewords and the received vector. The transmitter sends the noiseless confirmation when the largest information density corresponds to the true message. Averaging over all possible codebooks gives the achievability result. 

\subsection{Introducing Practical Constraints to VLFT}
\label{sec:MainConstribution} Define a VLFT code with the constraints of finite blocklength $N$ and uniform increment $I$ as follows:
\begin{definition}
\label{def:PracticalVLFT}
An $(\ell, M, N, I, \epsilon)$ VLFT code modifies 2) and 4) in Definition $1$ as follows: 
	\begin{enumerate} 
	\item[2')] A sequence of encoders 
	$$f_{n+kN}: \mc{U}\times\mc{W}\times\mc{Y}_{kN+1}^{n+kN-1}\mapsto \mc{X}, k \in \mb{N}$$
	that satisfies $f_n  = f_{n+N}$.
	\item[4')] A stopping time $\tau\in\{n_1 + kI:k\in\mb{N}\}$ w.r.t. the filtration $\mc{F}_n$ s.t. $\E[\tau] \leq \ell$ and $n_1$ is a given constant such that $n_1 + kI | N$ for some $k\in\mb{N}$.
	\end{enumerate}
\end{definition}

Define the fundamental limit of message cardinality $M$ for an $(\ell, M, N, I, \epsilon)$ VLFT code as follows:
\begin{definition}
Let $M^*_t(\ell, N, I, \epsilon)$ be the maximum integer $M$ such that there exists an $(\ell, M, N, I, \epsilon)$ VLFT code. For zero-error codes where $\epsilon = 0$, we denote the maximum $M$ as $M^*_t(\ell, N, I)$ and for zero-error codes with $I=1$ (i.e., decoding attempts after every received symbol) we denote the maximum $M$ as $M^*_t(\ell, N)$.
\end{definition}



For a feedback system that conveys $M$ messages with expected latency $\ell$, the expected throughput $R_t$ is given as $R_t = \log M / \ell$. All of the results that follow assume an arbitrary but fixed channel $\{P_{Y_j|X_j}\}_{j = 1}^N$ and a channel-input process $\{X_j\}_{j = 1}^{N}$ taking values in $\mc{X}$ where $N$ could be infinity. 

\subsection{The Finite-Blocklength Limitation}
\label{sec:VLFT_FiniteLength} 
This subsection investigates $(\ell, M, N, I, \epsilon)$ VLFT codes with finite $N$ but retains decoding at every symbol ($I=1$). The achievability results are examples of IR-NTC systems (or FV codes as described in Sec.~\ref{sec:Prelim}), so that encoding does not depend on the feedback except that feedback indicates when the transmission should be terminated. 

In an IR-NTC system, the expected latency $\E[\tau]$ is given as:
\begin{align}
\label{eqn:SumsOfPtau}
\E[\tau]  &= \sum_{n = 1}^{\infty} n \P[\tau = n]
\\
&= \sum_{n \geq 0} \P[\tau > n]
\\
\label{eqn:SumsOfJoints}
&=\sum_{n \geq 0} \P[E_n]
\\
\label{eqn:MarginalZeta}
&\le \sum_{n \geq 0} \P[\zeta_n] \, ,
\end{align}
where 
\begin{equation}
\label{eqn:JointErrorEvent}
E_n = \cap_{j = 1}^{n}\zeta_j
\end{equation} 
and $\zeta_j$ is the marginal error event at the decoder immediately after the $j$th symbol is transmitted.   Equation \eqref{eqn:MarginalZeta} follows since $E_n\subset \zeta_n$ implies $\P[E_n] \le \P[\zeta_n]$.

Consider a code $\mc{C}_N$ with finite blocklength $N$  and symbols from $\mc{X}$.  Achievability results for an $(\ell, M, N, 1, \epsilon)$ ``truncated'' VLFT code follow from a random-coding argument. In particular we have the following: 
\begin{theorem}
\label{thm:FiniteVLFT}
For any $M > 0$ there exists an $(\ell, M, N, 1, \epsilon)$ truncated VLFT code with
\begin{align}
\ell &\leq \sum_{n = 0}^{N-1} \xi_n
\\
\epsilon &\leq \xi_N.
\end{align}
where $\xi_n$ is the same as \eqref{eqn:Xi_n}. 
\end{theorem}

The proof is provided in Appendix~\ref{sec:AppendixVLFT}.
Achievability results for $\epsilon=0$ can be obtained using an $(\ell, M, N, 1, 0)$ ``repeated'' VLFT code, which repeats the transmission if the blocklength-$N$ codeword is exhausted without successful decoding. The transmission process starts from scratch in this case, discarding the previous received symbols.  Using the original $N$ symbols through, for example, Chase code combining would be beneficial but is not necessary for our achievability result. Specifically, for an $(\ell, M, N, I, 0)$ repeated VLFT code we have the following result:
\begin{theorem}
\label{thm:FiniteFV}
For every $M > 0$ there exists an $(\ell, M, N, 1, 0)$ repeated VLFT code such that
\begin{align}
\label{eqn:FiniteFV}
\ell &\leq \frac{1}{(1-\xi_N)} \sum_{n = 0}^{N-1}\xi_n,
\end{align}
where $\xi_n$ is the same as \eqref{eqn:Xi_n}.  
\end{theorem}

The proof is provided in Appendix~\ref{sec:AppendixVLFT}.
The rate penalty of using a finite-length (length-$N$) codebook is quantified in the following theorem and its corollary:
\begin{theorem}
\label{thm:VLFTExpandFinite}
For an $(\ell, M, N, 1, 0)$ repeated VLFT code with $N = \Omega(\log M)$, we have the following upper bound on $\ell$ for a stationary DMC with capacity $C$:
\begin{align}
\label{eqn:ellExpansion1}
\ell \leq \frac{\log M}{C} + c\, .
\end{align}
Let  $C_\Delta = C- \Delta$ for some $\Delta > 0$ and $N = \log M / C_\Delta$. The $O(1) $ term $c$ due to the finite value of $N$ is upper bounded by the following expression:
\begin{align}
\label{eqn:ellExpansion2}
c \le \frac{b_2 \log M}{C(M^{b_3/C_\Delta})} 
+ \frac{b_0\log M}{C_\Delta M^{b_1\Delta/C_\Delta}} + a
\end{align}
where $a$ depends on the mean and uniform bound of $i(X;Y)$, and $b_j$'s are constants related to $\Delta$ and $M$.  The proof is provided in Appendix~\ref{sec:AppendixVLFT}.
\end{theorem}

This choice of $N$ may have residual terms decaying with $M$ very slowly.  However, our non-asymptotic numerical results in Sec. \ref{sec:NumResults} for a BSC indicate that this decay is fast enough for excellent performance in the short-blocklength regime. 

An asymptotic expansion of $\log M_t^*(\ell, N)$ needs to be independent of $M$ and requires $N$ growing with $\ell$.   However, the components of the correction term $c$ in Thm. \ref{thm:VLFTExpandFinite} depend on both $N$ (as  $\log M / C_\Delta$) and $M$.  Indeed for a fixed $\ell$, the smallest $M$ satisfying \eqref{eqn:ellExpansion1} and \eqref{eqn:ellExpansion2} is achievable.  The argument we make below is that for any fixed constant $c=c_0 > 0$, there is an $\ell_0$ that depends logarithmically on $c_0^{-1}$ such that the expansion $\log M_t^*(\ell, N) \geq C \ell - c_0$ is true for all $\ell \geq \ell_0$. 
We first invoke the converse for an $(\ell, M, \infty, 1, 0)$ VLFT code:
\begin{theorem}[\cite{PolyIT11}, Thm. 11] \label{thm:poly11}
Given an arbitrary DMC with capacity $C$ we have the following for an $(\ell, M,\infty, 1, 0)$ VLFT code:
\begin{align}
\log M_t^* \leq \ell C + \log(\ell+1) + \log e\,.
\end{align}
\end{theorem}
After some manipulation, Thms. \ref{thm:VLFTExpandFinite} and \ref{thm:poly11} imply the following:
\begin{corollary}
\label{cor:ScalingForEll}
For an $(\ell, M, N, 1, 0)$ repeated VLFT code, we can pick $\delta > \frac{\log(\ell+1) + \log e}{C}$  and let $N = (1+\delta) \ell = \ell + \Omega(\log\ell)$ such that the following holds for a stationary DMC with capacity $C$:\footnote{As opposed to the expression in \cite{PolyIT11}, we use a minus sign for $O(1)$ term to make the penalty clear.}
\begin{align}
\log M_t^*(\ell, N) \geq \ell C - O(1)\,.
\end{align}
\end{corollary}
The proof is provided in Appendix~\ref{sec:AppendixVLFT}.

To conclude this discussion of the penalty associated with finite blocklength, we comment that $N$ only needs to be scaled properly, i.e., $N = \ell + \Omega(\log\ell)$, to obtain the infinite-blocklength expansion of $M^*_{t}(\ell, \infty)$ provided in \cite{PolyIT11}. Therefore the restriction to a finite blocklength $N$ does not restrict the asymptotic performance if $N$ is selected properly with respect to $\ell$.  The constant penalty terms in the expansion are different for infinite and finite $N$, which might not be negligible in the short-blocklength regime. Still, our numerical results in Sec.~\ref{sec:NumResults} indicate that relatively small values of $N$ can yield good results for short blocklengths.

\subsection{Limited, Regularly-Spaced, Decoding Attempts}
\label{sec:VLFT_TimeLimit1} 
This subsection investigates $(\ell, M, N, I, \epsilon)$ VLFT codes with $N=\infty$ but decoding attempted only at specified, regularly-spaced symbols ($I>1$).  The first decoding time occurs after $n_1$ symbols (which could be larger than $I$) so that the decoding attempts are made at the times $n_j = n_1 + (j - 1)I$.  The relevant information density process $i(X^{n_j}; Y^{n_j})$ is on the subsequence $n_j = n_1 + (j - 1)I$.  The main result here is that the constant penalty now scales linearly with $I$:
\begin{theorem}
\label{thm:VLFT_TimeLimit1}
For an $(\ell, M, N, I, 0)$ VLFT code with uniform increments $I$ and $N = \infty$ we have the following expansion for a stationary DMC with capacity $C$:
\begin{align}
\log M^*_t(\ell, \infty, I) \geq \ell C - O(I)\,.
\end{align}
The proof is provided in Appendix~\ref{sec:AppendixVLFT}.
\end{theorem}

In view of Thm. \ref{thm:VLFT_TimeLimit1}, the penalty is linear in the increment $I$.  The increment $I$ can grow slowly, e.g., $I = O(\log \ell)$ and still permit an expected rate that approaches $C$ without the dispersion penalty incurred when feedback is absence \cite{PolyIT10}. In the non-asymptotic regime, however, the increment $I$ must be carefully controlled to keep the penalty small.   Our numerical results in Sec.~\ref{sec:NumResults} indicate that choosing $I = \lceil\log_2 \log_2 M\rceil$ yields good results for short blocklengths.  Also, varying $I_j$ to decrease with $j$ can avoid a substantial penally while keeping the penalty small.

\subsection{Finite Blocklength and Limited Decoding Attempts}
\label{sec:VLFT_TimeLimit2} 
This subsection investigates $(\ell, M, N, I, 0)$ repeated VLFT codes with {\em both} finite $N$ and $I>1$.   When these two limitations  are combined, a key parameter is $m$, the number of decoding attempts before the transmission process must start from scratch if successful decoding has not yet been achieved. The parameters $m$, $n_1$, $N$ and $I$ are related by the following equation:
\begin{equation}
N=  n_1 + (m-1)I
\end{equation}

Combining the results of Sec.~\ref{sec:VLFT_FiniteLength} and Sec.~\ref{sec:VLFT_TimeLimit1} yields the following expansion for  $(\ell, M, N, I, 0)$ repeated VLFT codes on a stationary DMC with capacity $C$:
\begin{theorem}
\label{thm:MainAsympResult2}
For an $(\ell, M, N, I, 0)$ repeated VLFT  with $N = \ell + \Omega(\log\ell)$ on a stationary DMC with capacity $C$ and $\ell$ sufficiently large
\begin{align}
\label{eq:Thm2}
\log M^*_t(\ell, N, I) \geq \ell C - O(I)\, .
\end{align}
\end{theorem}
Specifically, if we choose $N > \ell + \frac{\log(\ell+1) + \log e}{C}$ and have decoding attempts separated by an increment $I$, then the expansion is the same as the case with $I>0$ and  $N = \infty$ so that the penalty term is linear in $I$. 
\begin{proof}
The proof is based on the following lemma that provides an upper bound on $\ell$ for an $(\ell, M, N, I, 0)$ repeated VLFT codes :
\begin{lemma}
\label{lem:MainAsympResult1}
For an $(\ell, M, N, I, 0)$ repeated VLFT code with $N = \Omega(\log M)$, we have the following upper bound on $\ell$ for a stationary DMC with capacity $C$:
\begin{align}
\ell&\leq (1+\P[\zeta_N])^{-1}\frac{\log M}{C} +  \P[\tau_0 \geq m] + O(I)
\\
&\leq \frac{\log M}{C} + O(I) \,,
\end{align}
where $\tau_0$ is the stopping time in terms of the number of decoding attempts up to and including the first success. 
\end{lemma}
The proof of the lemma is similar to Thm.~\ref{thm:VLFTExpandFinite} and the details are in Appendix~\ref{sec:AppendixVLFT}. 

The proof of Thm.~\ref{thm:MainAsympResult2} now follows. For an $(\ell, M, N, I, 0)$ repeated VLFT code, pick $N$ as follows:
\begin{align}
N = (1+\delta)\ell, \text{ where } \delta > \frac{\log(\ell+1)+\log e}{\ell C}\,.
\end{align}
The result follows by a similar argument as for Cor.~\ref{cor:ScalingForEll} and by observing that for the condition of $\delta$, $(1+\delta)\ell = \ell + \Omega(\log\ell)$. The restriction on the initial blocklength $n_1$ only makes a constant difference. 
\end{proof}

Our earlier work \cite{Chen_2011_ICC} used an intuitive choice of the finite blocklength: $N = (1+\delta)\ell$ for a fixed $\delta > 0$, i.e., a blocklength that is larger than the target expected latency by the fraction $\delta$. Theorem \ref{thm:MainAsympResult2} validates such choice in terms of the optimality of the asymptotic expansion (since $\delta\ell \in \Omega(\log\ell)$), and indicates that the required overhead $\delta$ is decreasing as the target expected latency grows.


\begin{figure}[htp]
	 \centering
    \includegraphics[width=0.5\textwidth]{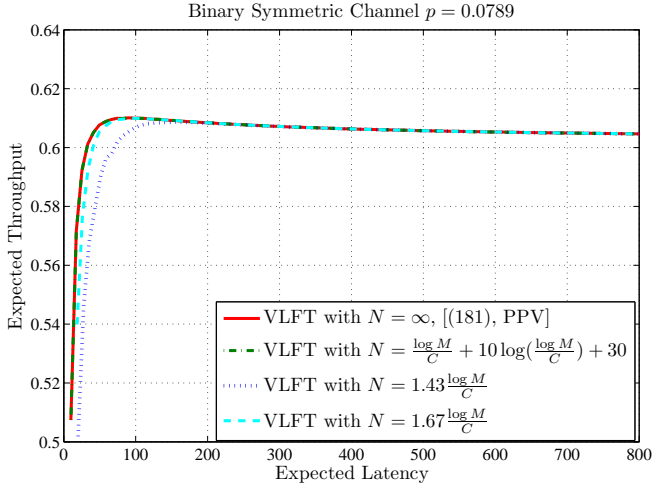}
    \caption{Performance comparison of VLFT code achievability based on the RCU bound with different codebook blocklengths.}
    \label{fig:VLFT_FiniteBlock}
\end{figure}

\subsection{Numerical Results}
\label{sec:NumResults}
For practical applications that apply feedback to obtain reduced latency, non-asymptotic behavior is critical.  This section gives numerical examples of non-asymptotic results for a BSC. For a BSC with transition probability $p$ we use the RCU bound in \cite{PolyIT10,PolyIT11}\footnote{We replace $(M-1)$ by $M$ for simplicity.}, which gives the following expression:
\begin{align*}
\xi_n \leq \sum_{t = 0}^{n} {n\choose t}p^t (1-p)^{n-t}\min\left\{1, M\sum_{j=0}^{t}{n\choose j} 2^{-n}\right\} \, .
\end{align*}

Fig.~\ref{fig:VLFT_FiniteBlock} shows expected throughput ($R_t$) vs. expected latency ($\ell$) performance of a VLFT code with $N=\infty$ and three finite-$N$ repeated VLFT codes  over a BSC with $p = 0.0789$.
Since $\ell$ scales linearly with $\frac{\log M}{C}$, for one repeated VLFT code $N$ scales as:
\begin{align}
N = \frac{\log M}{C} + a \log\left(\frac{\log M}{C}\right) + b\,,
\end{align}
so that $N = \ell + \Omega(\log\ell)$.  The constants $a, b > 0$ were selected experimentally to be  $a = 10$, $b = 30$. 

For the other two repeated VLFT codes, $N = \log M / C_\Delta$ where $C_\Delta = C - \Delta $.  We chose $\Delta = 0.3C$ and $0.4C$, which are $43\%$ and $67\%$ longer, respectively, than the blocklength $N = \log M / C$ that corresponds to capacity. In other words, $N = 1.43\log M/C$ and $N = 1.67\log M/C$ respectively. 

Expected throughput for the finite-$N$ repeated VLFT codes converges to that of VLFT with $N = \infty$ before expected latency has reached $200$ symbols. For expected latency above $75$ symbols, the only repeated VLFT code with visibly different throughput than the $N=\infty$ VLFT code is the $\Delta=0.3C$ code.

As mentioned in Sec.~\ref{sec:Prelim}, VLFT codes can have expected throughput higher than the original BSC capacity of $0.6017$ because of the NTC. This effect vanishes as expected latency increases.

\begin{figure}[htp]
	 \centering
    \includegraphics[width=0.5\textwidth]{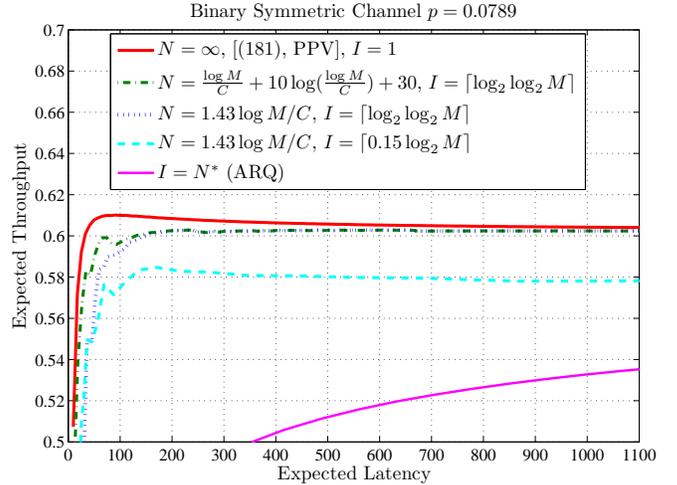}
    \caption{Performance comparison of VLFT code achievability based on the RCU bound with uniform increment and finite-length limitations.}
    \label{fig:VLFT_UnifInc}
\end{figure}

Fig.~\ref{fig:VLFT_UnifInc} shows the $R_t$ vs. $\ell$ performance of a VLFT code with $N=\infty$ and $I=1$ and repeated VLFT codes with various decoding-time increments $I$.  As in \eqref{eq:Thm2}, when $I$ grows linearly with $\log M$ (i.e., $\lceil 0.15\log_2 M\rceil$)  then there is a constant gap from the $I=1$ case.  However, if $I$ grows as $\lceil\log_2 \log_2 M\rceil$ then the gap from the $I=1$ case decreases as expected latency increases. ARQ performance (in which $I=N^*$, where $N^*$ is the  optimal blocklength for $M$) is also shown in the figure, which reveals a considerable performance gap from even the most constrained repeated VLFT implementation in Fig.~\ref{fig:VLFT_UnifInc}.

\section{IR-NTC with Convolutional and Turbo Codes}
\label{sec:VLFTexamples}
The achievability proofs in \cite{PolyIT11} and \ref{sec:VLFT} are based on an IR-NTC scheme using {\em random} codebooks. This section provides examples of IR-NTC based on {\em practical} codebooks: rate-compatible families of turbo codes and tail-biting convolutional codes.   The constraints of finite $N$ and $I>1$ studied analytically in the previous section appear naturally in the context of these codes. 

\subsection{Implementation of a repeated IR-NTC with Practical Codes}
\label{sec:SchemeIR}
A practical way to implement repeated IR-NTC is by using a family of rate-compatible codes with incremental blocklengths $\{n_i\}_{i = 1}^m$ where $n_i = \sum_{j = 1}^i I_j$. 
Defining an $(n, M)$ code to be a collection of $M$ length-$n$ vectors taking values in $\mc{X}$, we define a family of rate-compatible codes as follows:
\begin{definition}
\label{def:RCcode}
Let $n_1 < n_2 < \dots < n_m$ be integers. A collection of codes $\{\mc{C}_j\}_{j = 1}^m$ is said to be a family of rate-compatible codes if each $\mc{C}_j$ is an $(n_j,M)$ code that is the result of puncturing a common mother code, and all the symbols in the higher-rate code $\{\mc{C}_j\}$ are also in the lower rate code $\{\mc{C}_{j+1}\}$.
\end{definition}

A family of rate-compatible codes can be constructed by finding a collection of compatible puncturing patterns \cite{Hagenauer_Rate_1988} satisfying Def.~\ref{def:RCcode} for an $(N, M)$ mother code. Note that the puncturing becomes straightforward if we reorder the symbols of the mother code so that the symbols of $\mc{C}_1$ are first, followed by the symbols of $\mc{C}_2$ and so on. From this perspective, the symbols transmitted by VLFT in \cite{PolyIT11} can be seen as an infinite family of rate-compatible codes resulting from such an ordered puncturing.  

When implemented using a family of rate-compatible codes $\{\mc{C}_j\}_{j = 1}^m$, repeated IR-NTC works as follows.  A codeword of $\mc{C}_1$ with blocklength $n_1=I_1$ is transmitted to convey one of the $M$ messages. The decoding result is fed back to the transmitter and an NTC is sent if the decoding is successful. Otherwise the transmitter will send $I_2$ coded symbols such that the $n_2 = I_1 + I_2$ symbols form the codeword in $\mc{C}_2$ representing the same message. The decoder attempts to decode with code $\mc{C}_2$ and feeds back the decoding result. If decoding is not successful after the $m$th transmission where $n_m = N$, the decoder discards all of the previously received symbols and the process begins again with the transmitter resending the $I_1$ initial coded symbols. This repetition process continues until the decoding is successful. 



In the special case where $m=1$  the repeated IR-NTC reduces to SCR-NTC, which we refer to as ARQ.
\subsection{Randomly Punctured Convolutional and Turbo Codes}
\label{sec:SimSetup}
The practical examples of repeated IR-NTC provided in this paper use tail-biting RCPC codes and RCPT codes. The details of the rate-compatible codes used in this paper are given as follows: 


The two convolutional codes we used in this paper are a $64$-state code and a $1024$-state code with generator polynomials $(g_1, g_2, g_3) = (133, 171, 165)$ and $(2325,2731,3747)$ in octal, respectively. The  $64$-state code is from the 3GPP-LTE \cite{3GPP} standard and the $1024$-state code is the optimal free distance code from \cite[Table 12.1b]{BookLin}. Both of the codes are implemented as tail-biting codes \cite{Ma_86} to avoid rate loss at short blocklengths.

The turbo code used in this paper is from the 3GPP-LTE standard \cite{3GPP}, i.e., the turbo code with generator polynomials $(g_1, g_2) = (13, 15)$ in octal, with a quadratic permutation polynomial interleaver.

Pseudo-random puncturing, also referred to as circular buffer rate matching in \cite{3GPP}, provides the rate-compatible families for both the convolutional and the turbo codes. The process is shown in Fig.~\ref{fig:CBRM}: the encoder first generates a rate-$1/3$ codeword. Then the output of each of the encoder's three bit streams passes through a ``sub-block'' interleaver with a blocklength $K$. The interleaved bits of each sub-block are concatenated in a buffer, and bits are transmitted sequentially from the buffer to produce the increments $I_j$ as shown in Fig.~\ref{fig:BitCollection}.   The sub-block interleavers re-order the bits of the mother code so that sequential transmission of the bits produces an effective family of rate-compatible codes as discussed in Sec.~\ref{sec:SchemeIR}.

\begin{figure}[t]
\centering
\def\svgwidth{0.4\textwidth}
   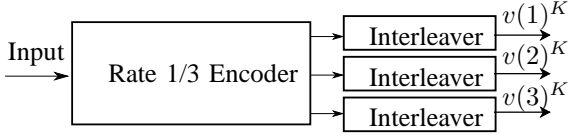
\caption{Pseudo-random puncturing (or circular buffer rate matching) of a convolutional code. At the bit selection block, a proper amount of coded bits are selected to match the desired code rate.}
\label{fig:CBRM}
\end{figure}
\begin{figure}[t]
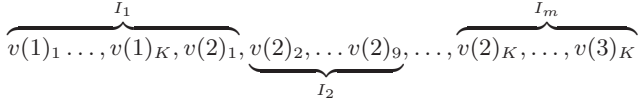

{\small
\begin{align*}
\overbrace{v(1)_1 \dots, v(1)_{K},v(2)_1}^{I_1}, \underbrace{v(2)_2,\dots v(2)_9}_{I_2}, \dots, \overbrace{v(2)_{K}, \dots, v(3)_{K}}^{I_m}
\end{align*}
}
\caption{Illustration of an example of transmitted blocks for rate-compatible punctured convolutional codes}
\label{fig:BitCollection}
\end{figure}
We will use simulation results of repeated IR-NTC systems based on these RCPC and RCPT codes to compare with our analysis in the following sections. 

\section{Rate-Compatible Sphere-Packing (RCSP)}
\label{sec:IR_RCU_RCSP}
The random-coding approach gives a tight achievable bound on expected latency when $M$ is sufficiently large. In the short-latency regime, however, practical codes can outperform random codes, as noted in \cite{Williamson_ISIT_2012}. To find a code-independent analysis that gives a better prediction of practical code performance, we introduce the rate-compatible sphere-packing (RCSP) approximation.  As we will see in Sec.~\ref{sec:Optimization}, RCSP can also facilitate the optimization of the increment lengths $I_j$ and provide a trajectory of target error rates for use in the design of a rate-compatible code family.

Shannon et al. \cite{Shannon_1967} derived lower bounds of channel codes for DMC by packing typical sets into the output space. The typical sets are related to the divergence of the channel and an auxiliary distribution on the output alphabet.  For the AWGN channel, Shannon \cite{Shannon_1959} showed both the lower bounds on the error probability by considering optimal codes on a sphere (the surface of the relevant ball). The bound turns out to be tight even in the finite-blocklength regime as indicated in \cite{PolyIT10}. One drawback for considering codes on a sphere is the computational difficulty involved, even for a single fixed-length code.

RCSP is an approximation of the performance of repeated IR-NTC using a family of rate-compatible codes. The idea of RCSP is an extension of the sphere-packing lower bound from a single fixed-length code to a family of rate-compatible codes. For the ideal family of rate-compatible codes, each code in the family would achieve perfect packing.  Our analysis will involve two types of packing: 1) perfect packing throughout the volume of the ball whose radius is determined from the signal and noise powers or 2) perfect packing on the surface of the ball whose radius is determined by the signal power constraint.  We will also consider both maximum-likelihood (ML) decoding and bounded-distance (BD) decoding.

Let $\{\mc{C}_j\}_{j = 1}^m$ be a family of rate-compatible codes. Let the marginal error event of the code $\mc{C}_j$ at blocklength $n_j$ be $\zeta_{n_j}$ and let the joint error probabilities $\P[E_{n_j}]$ be defined similar to \eqref{eqn:JointErrorEvent}: 
\begin{equation}
E_{n_j} = \cap_{i = 1}^j \zeta_{n_i}\,.
\end{equation}
The expected latency for a repeated IR-NTC can be computed as follows:
\begin{align}
 \label{eqn:EnLatency}
\ell &= \frac{I_1 + \sum_{j = 2}^{m} I_j\P[E_{n_{j-1}}]}{1 - \P[E_{n_m}]}.
\end{align}

Applying the ideal of RCSP throughout the volume of the ball to \eqref{eqn:EnLatency} leads to the joint RCSP approximation of the repeated IR-NTC performance.  Since $\P[\zeta_{n_j}] \ge\P[E_{n_j}]$, replacing $\P[E_{n_j}]$ with $\P[\zeta_{n_j}]$ produces a upper bound on expected latency as follows:
\begin{align}
\label{eqn:ZetaLatency}
\ell & \leq \frac{I_1 + \sum_{j = 2}^{m} I_j\P[\zeta_{n_{j-1}}]}{1 - \P[\zeta_{n_m}]}.
\end{align}
Since $\P[\zeta_{n_j}]$ is often a tight upper bound on $\P[E_{n_j}]$ (examples will be shown in Sec.~\ref{sec:ChernoffBounds}), applying the ideal of RCSP throughout the volume of the ball to \eqref{eqn:ZetaLatency} produces the marginal RCSP approximation of the expected latency in a repeated IR-NTC, which is very similar to the joint RCSP approximation and more easily computed.  

Performing ML decoding gives a lower bound on the expected latency for the repeated IR-NTC, but makes the error probabilities difficult to evaluate. We initially analyze using BD decoding to decode the ideal code family $\{\mc{C}_j\}_{j = 1}^m$.   Subsequently we refine the analysis by bounding ML decoding performance.  Thus we will be considering both the joint and marginal RCSP approximations and also considering both BD and ML decoding.

\subsection{Marginal RCSP Approximation for BSC}
\label{sec:RCSP_BSC}

For the BSC the optimal decoding regions are simply Hamming spheres.   RCSP upper bounds IR-NTC performance on the BSC by assuming that each code in the family of rate-compatible codes achieves the relevant Hamming bound.  

 For a BSC with transition probability $p$ the marginal error probability $\P[\zeta_{n_j}]$ is lower bounded as follows:\footnote{Formally, the sphere-packing lower bound for BSC follows from the fact that the tail of a binomial r.v. is convex.}    
\begin{equation}
\label{eqn:BSC_RCSP}
\P[\zeta_{n_j}] \geq \sum_{t = r_{j}+1}^{n_j} {n_j\choose t}p^t(1-p)^{n-t},
\end{equation}
where $r_{j}$ is chosen such that  
\begin{equation}
M\sum_{t=0}^{r_j - 1}{n_j \choose t} + \sum_{i = 1}^{M}A_{i} = 2^{n_j}
\end{equation}
and
\begin{equation}
0 < \sum_{j = 1}^{M} A_i < M{n_j \choose r_j}.
\end{equation}
We use \eqref{eqn:BSC_RCSP} to compute the marginal RCSP approximation for the BSC. Note that for the BSC and an $(n, M)$ code with $M$ uniform decoding regions that perfectly fill the space, ML decoding is also BD decoding with an uniform radius.

For the BSC with $p=0.0789$, Fig.~\ref{fig:BSC_Conv_Eg} shows the marginal RCSP approximation with $m = \infty$, the VLFT converse of \cite{PolyIT11}, the random-coding achievability of \cite{PolyIT11}, and simulations of repeated IR-NTC using the $64$-state convolutional code from Sec.~\ref{sec:SimSetup}. All simulation points have increment $I = 1$, finite codeword lengths $N=3k$, and initial blocklength $n_1 = k$ for $k = 16, 20, 32, 64$, respectively. To compare the choice of $N=3k$ with the results in Sec.~\ref{sec:VLFT}, take $k = 32$ as an example. For $k = 32$ the expected latency is $\ell = 49.62$ in Fig.~\ref{fig:BSC_Conv_Eg} and $N = 96$ corresponds to choosing a $\delta = 0.93$.  As we saw in Fig. \ref{fig:VLFT_FiniteBlock}, $\delta=0.67$ gives RCU performance indistinguishable from $N=\infty$ so that $N=96$ should be more than sufficient for this example.

The convolutional code simulations give throughput-latency points that are very close to the marginal RCSP approximation for expected latency less than $50$. The simulation points for $k = 16, 20, 32$ are significantly higher than the random-coding achievability result of \cite{PolyIT11}. The simulation point for $k = 64$ falls below the random-coding achievability because the expected latency is larger than the analytic trace-back depth of $60$ bits (or 20 trellis state transitions) for the $64$-state convolutional code.

Note that in Fig.~\ref{fig:BSC_Conv_Eg}, the convolutional code simulation points, the marginal RCSP approximation curve, the VLFT achievability, and the VLFT converse curves are all above the BSC capacity because the NTC assumption provides additional information to the receiver as discussed in Sec.~\ref{sec:Prelim}.
\begin{figure}[t]
\centering
   \includegraphics[width=0.5\textwidth]{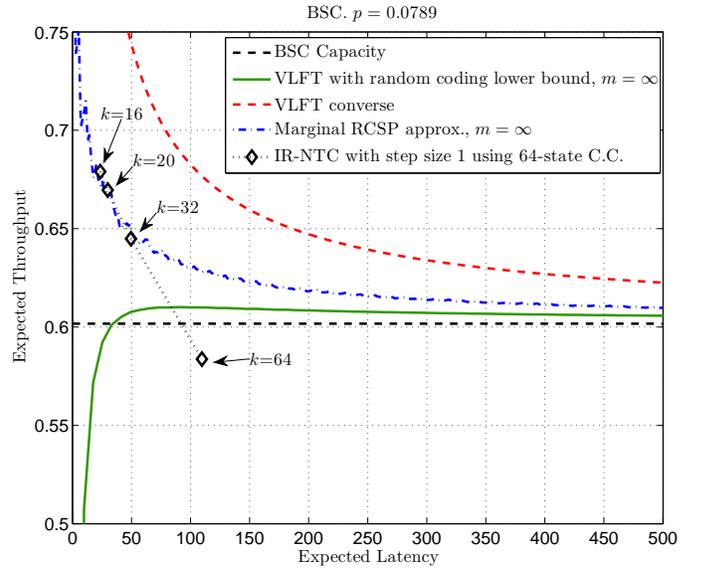}
\caption{In the short-latency regime, the marginal RCSP approximation can be more accurate for characterizing the performance of good codes (such as the convolutional code in this figure) than random-coding analysis. The additional information provided by the error-free termination symbol of NTC leads to a converse and an operational rate for the convolutional code that are above the original BSC capacity.}
\label{fig:BSC_Conv_Eg}
\end{figure}

\subsection{RCSP Approximation for AWGN Channel}
\label{sec:RCSP_AWGN}
This subsection derives joint and marginal RCSP approximations for the AWGN channel under both BD and ML decoding. Consider an AWGN channel $Y=X+Z$ with an average power constraint $P$.\footnote{We use the expectation average power constraint $\E\left(\sum_{j=1}^n X_j^2\right) \leq nP $, which is not as strict as the summation average power constraint $\sum_{j=1}^n X_j^2 \leq nP$ since it allows the codeword power be larger than $P$ with low probability. As shown in the Appendix~\ref{sec:PowerConstraintRCSP}, the sphere-packing property combined with the expectation average power constraint will satisfy the summation average power constraint if the rate is less than capacity.} 
Let the signal-to-noise ratio (SNR) be $\eta = P/\sigma^2$ where $P$ is the signal power and $\sigma^2$ is the noise power. Assume without loss of generality (w.l.o.g.) that each noise sample has a unit variance. The average power of a length-$n$ received word is $\E[ \|Y^n\|^2] \leq n(P + \sigma^2) = n(1 + \eta)$.  Sphere packing seeks a codebook that has $M$ equally separated codewords within the $n$-dimensional norm ball with radius $r_{\text{outer}} = \sqrt{n(1 + \eta)}$.

One can visualize a large outer sphere that contains $M = 2^{nR}$ decoding regions, $D_i, i = 1, \dots, M$, each with the same volume. Considering the spherical symmetry of i.i.d. Gaussian noise $Z$, our bounded-distance decoding region is an $n$-dimensional Euclidean ball centered around the message point. By conservation of volume, the largest radius of the decoding region has a volume satisfying
\begin{align*}
 \text{Vol}(D_i) &= K_n   r^n  
 \\
 &\leq 2^{-nR}\times{ \text{Vol}({\text{Outer sphere}})} 
 \\
 &= 2^{-nR}{K_n \left( {\sqrt {n(1 + \eta )} } \right)^n }
\end{align*}
where $K_n$ is the volume of the $n$-dimensional unit sphere. Solving for the radius of the decoding region yields
\begin{align}
\label{eq:r_inner}
r \leq  2^{-R}\sqrt {n(1 + \eta )}.
\end{align}

\subsubsection {RCSP Approximation for AWGN under BD decoding}
Using bounded-distance decoding on the AWGN channel gives the following decoding rule: the decoder selects message $j$ if the received word $Y^n \in D_i$ for a unique $i\in\{ 1, \dots, M\}$ and declares an error otherwise. Regardless of the transmitted codeword, the marginal error event for BD decoding is $\zeta_{n} = \{Z^{n}: \|Z^{n}\|^2 > r^2\}$ where $r$ is the decoding radius. The error probability for decoding region with radius $r$ is then simply given by the tail of a chi-square r.v. with $n$ degrees of freedom:
\begin{align}
\P[\zeta_{n}] &= 1-F_{\chi_n^2}(r^2)
\\
\label{eq:chi-square-tail}
&= G_{\chi_n^2}(r^2) 
\end{align}
where $F_{\chi_n^2}$ is the CDF of a chi-square distribution with $n$ degrees of freedom. 

Let $M = 2^k$ for an integer $k$. Assuming perfect sphere-packing (i.e., achieving \eqref{eq:r_inner} with equality) at each incremental code rate, the radius at incremental code rate $k/n_j$ is given as:
\begin{align}
\label{eqn:r_i}
r_j^2 = \frac{n_j(1+\eta)}{2^{2k/n_j}}\,.
\end{align}

Note that the probability of a decoding error in the $j$th transmission depends on previous error events. Conditioning on previous decoding errors $\zeta_{n_j}, j = 1\dots, m-1$ makes the error event  $\zeta_{n_m}$ more likely than the marginal distribution would suggest. 
Recall that $E_{n_m}$ denotes $\cap_{j \leq m} \zeta_{n_j}$; the joint probability $\P[E_{n_j}], 1\le j\le m$ is given as:
\begin{align}
	\label{eqn:Pzeta_i}
	\nonumber
	\P[E_{n_j}] = \int_{r_1^2}^{\infty} &\int_{r_2^2-t_1}^{\infty} \dots \int_{r_{j-1}^2- 
	\sum_{i=1}^{j-2}t_i}^{\infty}  G_{\chi_{I_{j}}^2} \left (r_j^2 - \sum_{i=1}^{j-1}t_i\right ) 
	\\
	 &dF_{\chi_{I_1}^2}(t_1) \dots dF_{\chi_{I_{j-1}}^2}(t_{j-1}) ,
\end{align}
where the increments $I_{j}$ are as defined in Sec. \ref{sec:SchemeIR} and $G_{\chi_{I_{j}}^2}$ is defined in \eqref{eq:chi-square-tail}.

Using  $\P[E_{n_j}]$ or $\P[\zeta_{n_j}]$ as derived above, we can compute the expected latency $\ell$ of the repeated IR-NTC by using the joint BD RCSP approximation as in \eqref{eqn:EnLatency} or the marginal BD RCSP approximation as in \eqref{eqn:ZetaLatency}. The expected throughput $R_t$ is given by $k/\ell$.

\begin{figure}[t]
\centering
   \includegraphics[width=0.5\textwidth]{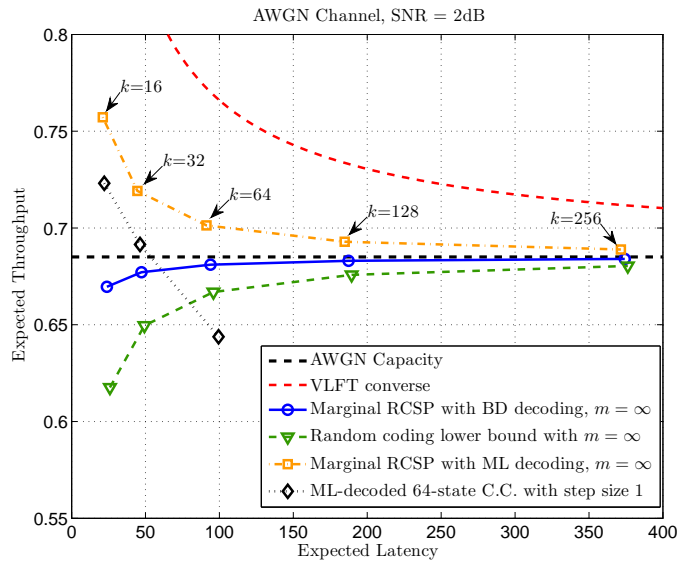}
\caption{The marginal RCSP approximation with BD decoding, and marginal RCSP with ML decoding, for the AWGN channel. Similar to the BSC, the additional information provided by the error-free termination symbol of NTC leads to a converse and an operational rate for the convolutional code that are above the original BSC capacity.}
\label{fig:AWGN_RCSP_VLFT_mInfty}
\end{figure}
\subsubsection{Marginal RCSP Approximation under ML decoding}

To compute the marginal RCSP approximation with ML-decoding, we apply Shannon's sphere-packing lower bound\footnote{We bound the error probability using Shannon's argument: the optimal error probability for 
codewords inside an $n$-dimensional ball with radius $\sqrt{nP}$ is lower bounded by the optimal error 
probability for codewords on an $n+1$-dimensional sphere surface with radius $\sqrt{(n+1)P}$.} \cite{Shannon_1967}. We use the asymptotic approximation in \cite{Shannon_1967} to estimate the lower bounds of the marginal error probabilities. The approximation formula in \cite{Shannon_1967} does not work well for rates above capacity, providing negative values that are trivial lower bounds. To obtain a good estimate in these cases we replace those trivial probabilities with the ones suggested by the marginal RCSP approximation using BD decoding. 

For the $2$dB AWGN channel, Fig. \ref{fig:AWGN_RCSP_VLFT_mInfty} shows the VLFT converse of  \cite{PolyIT11}, the random-coding lower bound of  \cite{PolyIT11}, the marginal RCSP with BD decoding, and marginal RCSP with ML decoding, all with $m = \infty$ and $I=1$. To approximate the behavior of $m=\infty$ we found numerically that $m = 10k$ is sufficient in this example. Also shown in Fig. \ref{fig:AWGN_RCSP_VLFT_mInfty} is the simulation of the repeated IR-NTC with step size $I=1$ and $k = 16,~32,$ and $64$ using the $64$-state convolutional code from Sec.~\ref{sec:SimSetup}. 

Similar to the BSC example, Fig.~\ref{fig:AWGN_RCSP_VLFT_mInfty} shows that the VLFT converse based on Fano's inequality \cite{PolyIT11} and the marginal RCSP approximation with ML decoding are both above the asymptotic capacity due to the information provided by NTC. The marginal RCSP approximation with BD decoding is not as accurate as the curve with ML decoding in the short-latency regime. The difference between the two curves, however, becomes small beyond the expected latency of $200$.

For $k = 16$ and $32$, the convolutional code simulations of repeated IR-NTC (even with a finite $N$ of $3k$) with binary modulation outperform the corresponding points on the curve for VLFT random-coding achievability  with $N = \infty$ and unconstrained input to the channel. The two simulation points generally follow the curve of the marginal RCSP approximation with ML decoding, but with a gap. We suspect that this gap is due to the use of binary modulation rather than an unconstrained input in the simulation. As with the BSC example, for $k = 64$ the performance of the simulation falls short as the expected latency goes beyond the analytic trace-back depth of the $64$-state convolutional code. 


\begin{figure}[t]
\centering
    \includegraphics[width=0.5\textwidth]{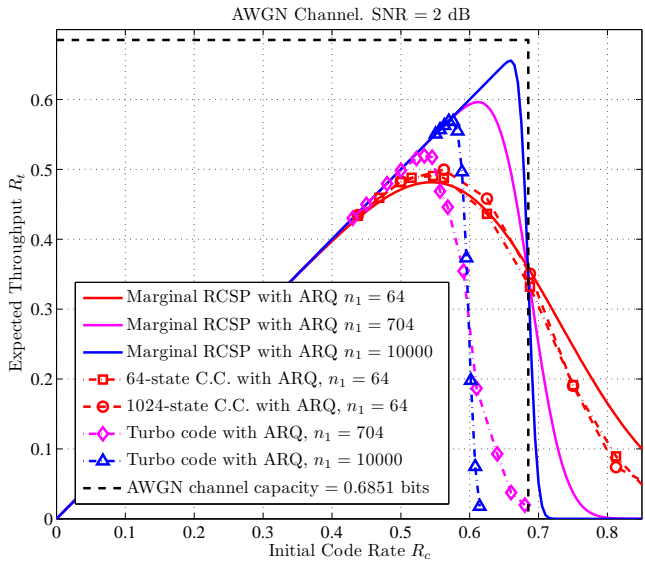}
    \caption{ARQ for RCSP, turbo code, $64$-state and $1024$-state convolutional codes simulations with different initial blocklengths.}
    \label{fig:Rc_Rt_CC_TC}
\end{figure}

\subsection{Marginal BD RCSP for ARQ and IR-NTC in AWGN}
\label{sec:MRCSP_Example}
This subsection provides examples of applying the marginal RCSP approximation using BD decoding for the AWGN channel. Note that all examples use BD decoding and we will simply use the term ``marginal RCSP approximation''. 
\subsubsection{ARQ}
\label{sec:SimpleARQ}
Consider ARQ (SCR-NTC) on an AWGN channel. Based on the marginal RCSP approximation the expected latency $\ell$ is given as
\begin{align}
	\ell &= \frac{n_1}{1-\P[\zeta_{n_1}]} 
	\\
	&= \frac{n_1}{F_{\chi_{n_1}^2}( r_1^2 )}\,,	
\end{align}
and expected throughput $R_t = k/\ell$ is given as
\begin{equation}
\label{eqn:Rt_ARQ}
	R_{t,\text{ARQ}} = R_c F_{\chi^2_{n_1}}(r_1^2)
\end{equation}
where $R_c = k/n_1$ is the initial code rate and $r_1$ is the decoding radius of the codeword with length $n_1$.  Note that in the case of ARQ (SCR-NTC), the joint RCSP and marginal RCSP approximations are identical since there is only a single transmission before repetition.

\begin{figure}[t]
\centering
    \includegraphics[width=0.5\textwidth]{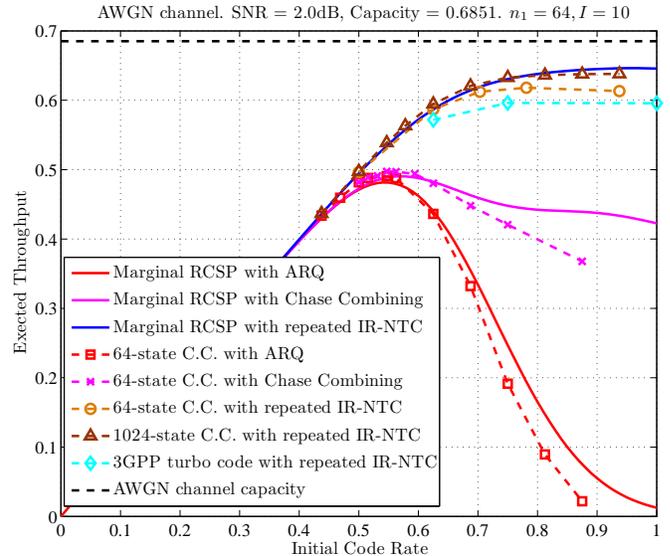}
    \caption{$R_t$ vs. $R_c$ for ARQ, ARQ with Chase combining and IR-NTC with $n_1 = 64$.}
    \label{fig:SARQ_Chase_IR}
\end{figure}

	Fig.~\ref{fig:Rc_Rt_CC_TC} shows the $R_t$ vs. $R_c$ (expected throughput vs. initial code rate) curve for the RCSP with ARQ. The SNR is $2$ dB and the initial blocklengths are $n_1=64$, $n_1=704$ and $n_1=10,000$. For each curve in Fig.~\ref{fig:Rc_Rt_CC_TC}, the initial blocklength $n_1$ is fixed as we vary the number of messages $M = 2^k$ and the initial code rate $R_c = k/n_1$ changes accordingly. Note that for ARQ, the initial blocklength is also the length of the possible subsequent transmissions. 
	
	For blocklengths $n_1=704$ and $n_1=10000$, Fig.~\ref{fig:Rc_Rt_CC_TC} compares RCSP with the 3GPP-LTE turbo codes.  Each point of the dash-dot turbo-code curves represents a different turbo code with the same blocklength but different code rate (different $k$). Interestingly, after the initial negligible-codeword-error region where expected throughput equals code rate, the RCSP curve and the turbo code curve are very similar for both $n_1=704$ and $n_1=10,000$.  Furthermore, in this region the difference between the code rate $R_c$ associated with a given throughput for RCSP and for the turbo code is about $0.1$ bits for both $n_1=704$ and $n_1=10,000$ despite the large difference between these two blocklengths.
	
	For blocklength $n_1=64$, Fig.~\ref{fig:Rc_Rt_CC_TC} compares the marginal RCSP approximation to the performance of $64$-state and $1024$-state tail-biting convolutional codes. The marginal RCSP approximation closely predicts the performance of these good tail-biting convolutional codes.  Thus, while the turbo codes achieve higher throughputs than the convolutional codes, the convolutional codes perform closer to the RCSP approximation for their (short) blocklengths than do the turbo codes relative to the RCSP approximation for their (longer) blocklengths.

RCSP with its optimistic decoding radius $r_1$ and suboptimal bounded-distance decoding is a mixture of optimistic and pessimistic assumptions.  These plots, however, show that RCSP can provide accurate guidance for good short-blocklength codes such as tail-biting convolutional codes.

\subsubsection{Chase Combining and IR-NTC}
\label{sec:FixedIR}
	In the ARQ scheme, the same complete codeword is transmitted at each retransmission. One way to utilize these repetitions is to apply the combining scheme proposed by Chase \cite{Chase_1985}. The Chase scheme uses maximal ratio combining of $L$ repeated codewords at the receiver. If the $L$ codewords are transmitted directly through the AWGN channel with the same SNR, the combining process increases the effective SNR at the receiver by a factor of $L$. 
	
The side information that the previous block was decoded unsuccessfully (available from the absence of the NTC) implies that the instantaneous noise power in the previous packet is larger than the expected noise power. The error probability is lower bounded (and the throughput is upper bounded) by ignoring this side information. 
		
As shown in Fig.~\ref{fig:SARQ_Chase_IR}, Chase combining of all received packets for a message does not significantly increase the highest possible throughput.  Note that the RCSP curve for Chase combining uses the throughput upper bound that ignores the side information of previous decoding failures. Chase combining does provide a substantial throughput improvement for higher-than-optimal initial code rates, but these improved throughputs could have also been achieved by using a lower initial code rate and no Chase combining.

We now present an RCSP approximation and code simulations of repeated IR-NTC as described in Sec.~\ref{sec:SchemeIR}. 
The exact computation of the joint RCSP approximation is challenging when $m$ is large. We use the marginal RCSP approximation here to compute the $R_t$ vs. $R_c$ curve for the repeated IR-NTC. 

%
%
%
%
%

Fig.~\ref{fig:SARQ_Chase_IR} shows an $R_t$ vs. $R_c$ curve computed based on the approximation of RCSP for the repeated IR-NTC scheme with $n_1=64$, $m=10$ and a uniform increment $I=10$. We use a uniform increment to avoid the need to optimize each increment although the computation also admits non-uniform increments. We choose $m = 10$ because experimental results show that for $m>10$ the throughput improvement is diminishing. Note that we include AWGN channel capacity in Fig.~\ref{fig:SARQ_Chase_IR} only as a point of reference since finite-latency capacity with NTC is higher than the asymptotic capacity as shown in Sec.~\ref{sec:RCSP_AWGN}.

Fig.~\ref{fig:SARQ_Chase_IR} also shows repeated IR-NTC simulations of two tail-biting convolutional codes and a turbo code, where the relevant codes are described in Sec.~\ref{sec:SimSetup}.   The simulated turbo code saturates at a lower throughput than the $64$-state and $1024$-state tail-biting convolutional codes, which are both ML decoded.  We expect that turbo codes with better performance at short latencies can be found.  Still, the convolutional code performance is outstanding in this short-latency regime.

Fig.~\ref{fig:SARQ_Chase_IR} shows that ARQ achieves a throughput of less than $0.5$ bits and Chase combining provides little improvement over ARQ, with less than $0.01$ bits increase in expected throughput.  In contrast, the $1024$-state convolutional code simulation of repeated IR-NTC achieves a throughput of $0.638$ bits. This coincides with the observation in \cite{PolyIT11} that IR is essential in achieving high throughput with feedback in the finite-blocklength regime, even though the SCR scheme proposed by Yamamoto and Itoh~\cite{Yamamoto_1979} achieves the optimal error-exponent shown by Burnashev~\cite{Burnashev_1976}. 

We conclude this section by comparing Fig.~\ref{fig:Rc_Rt_CC_TC} and Fig.~\ref{fig:SARQ_Chase_IR}. As shown in Fig.~\ref{fig:Rc_Rt_CC_TC}, the expected throughput of an ARQ using the blocklength-$10,000$ turbo code is $0.56$ bits with an expected latency of $10,016$ bits. Fig.~\ref{fig:SARQ_Chase_IR} shows that the repeated IR-NTC using the $1024$-state convolutional code achieves a higher expected throughput of $0.64$ bits with an expected latency of only $100$ bits!

\section{Optimization of Increments}
\label{sec:Optimization}
Sec.~\ref{sec:MRCSP_Example} studied throughput by fixing the initial blocklength $n_1$ and varying the number of information symbols $k$, which correspondingly varied the initial code rate $R_c$.  This produces curves of expected throughput $R_t$ vs. initial code rate $R_c$.  In contrast, this section fixes $k$ and studies throughput by optimizing the set of increments $\{I_j\}_{j = 1}^m$ for the repeated IR-NTC scheme as presented in Sec.~\ref{sec:SchemeIR}. This produces curves of expected throughput $R_t$ vs. expected latency $\ell$. The optimization uses BD decoding to compute the expected throughput for both joint and marginal RCSP approximation. For  the rest of the paper we will assume BD decoding unless otherwise stated.  However, numerical results for ML decoding are also presented and compared to BD decoding. 

We first consider exact computations of the joint RCSP approximation for relatively small values of $m$. We begin with the special case of $m=1$, which is ARQ. We then provide results for cases where $m>1$ and study how increasing $m$ improves performance.    We introduce Chernoff bounds on the joint RCSP approximation and compare these bounds with the marginal RCSP approximation.  We then use the marginal RCSP approximation to optimize the performance of repeated IR-NTC for large $m$ but constrained to have uniform increments after the initial transmission.  Then we optimize increments for non-repeating IR-NTC constraining both the maximum number of incremental transmissions and the probability of outage.  Finally, we introduce the concept of the decoding error trajectory, which provides the RCSP approximation of error probability at each incremental transmission.  It is a useful guide for rate-compatible code design for feedback systems.

\subsection{Choosing $I_1$ for the $m=1$ Case (ARQ)}
\label{sec:optARQ}

%

Recall that for repeated IR-NTC the special case of $m=1$ is ARQ. In this case, when the message is fixed to be $k$ bits, the RCSP approximation\footnote{Note that for the $m=1$ case there is no distinction between the joint RCSP approximation and the marginal RCSP approximation.} of expected throughput is a quasi-concave function of the code rate $R_c$ in \eqref{eqn:Rt_ARQ}.  Thus a unique optimal code rate $R_c^*$ for the repeated codewords can be found numerically \cite{Boyd_2004_CO} to maximize the RCSP approximation of $R_t$ for a given $k$.   

Table~\ref{tbl:ARQ} presents the optimal code rates $R_c^*=k/I_1^*$ and transmission lengths $n_1=I_1^*$ for ARQ.  Fig.~\ref{fig:latVthroughput} plots the maximum RCSP approximation of throughput vs. expected latency $\ell$ for ARQ as the red  curve with diamond markers.  Both Table~\ref{tbl:ARQ} and Fig.~\ref{fig:latVthroughput} apply the constraint that the lengths $I_1^*$ must be integers. These results for ARQ are discussed together with the results of $m>1$ in the next subsection.

\begin{table}[t]
\centering
  \caption{Optimized  transmission lengths $n_1=I_1^*$ and initial code rates $R_c^*$ for ideal-sphere-packing ARQ  with information lengths $k$ and SNR $\eta = 2$dB.} 
\begin{tabular}{ c | c  c  c  c  c }
  $k$ & $16$ & $32$ & $64$ & $128$ & $256$ \\
  \hline
  $n_1=I_1^*$ &  $31$ &  $60$ &  $116$ &  $222$ & $429$ \\
  \hline
  $R_c^*$ & $0.516$   & $0.533$ &  $0.552$   & $0.577$ & $0.597$ \\
  \end{tabular}
\label{tbl:ARQ}
\end{table}

\subsection{Choosing Increments $\{I_j\}$ to Maximize Throughput}
\label{sec:OptimizeNonUniformI}

In Sec.~\ref{sec:FixedIR} we demonstrated one repeated IR-NTC scheme with $m=10$ transmissions that could approach capacity with low latency based on RCSP.  Specifically, the transmission lengths were fixed to $I_1=64$ and $I_2, \ldots, I_{10}=10$, while $k$ was varied to maximize throughput. 

This subsection presents optimization results based on exact numerical integrations computing the joint RCSP approximation. Both $k$ and the number of transmissions $m$ are fixed, and a search identifies the set of transmission lengths $\{I_j\}_{j = 1}^m$ that maximizes the joint RCSP approximation of expected throughput using BD decoding. Based on the optimized increments, this subsection also provides the marginal RCSP approximation of the expected throughput using ML decoding.

For $m>1$, identifying the transmission lengths $I_{j}$ which minimize the latency $\ell$ is not straightforward due to the joint decoding error probabilities in \eqref{eqn:Pzeta_i}.  Restricting to a small $m$ allows exact computation of  \eqref{eqn:Pzeta_i} in Mathematica, avoiding the marginal approximation. We study the cases when $m\le 6$ based on numerical integration. 

The computational complexity of numerical integration for \eqref{eqn:Pzeta_i} increases with the transmission index $j$.  Because of this increasing complexity we limited attention to a well-chosen subset of possible transmission lengths.  Thus our results based on numerical integration may be considered as lower bounds to what is possible with a fully exhaustive optimization. 

\begin{table}[t]
\centering
 \caption{Optimized RCSP transmission lengths for $m=5$ and SNR$=2$ dB using non-uniform increments. } 
\begin{tabular}{ c | c  c  c  c  c | c }
  $k$ & $I_1$ & $I_3$ & $I_3$ & $I_4$ & $I_5$ & $R_t$ Opt. \\
  \hline
  16 &   19 & 4 &    4     & 4     & 8  & 0.6019\\
  \hline
  32 &   38 &    8 &    8  &   8 &   12 & 0.6208 \\
  \hline
    64   & 85 &   12   &  8  &  12   & 16 & 0.6363 \\
  \hline
   128  & 176  &  14  &  14   & 14  &  28 & 0.6494\\
  \hline
   256 &  352  &    24 &   24  &  24  &  48 &0.6593\\
\end{tabular}
 \label{table:m5_steps}
\end{table}
Table~\ref{table:m5_steps} shows the results of the $m=5$ optimization (i.e., the set of lengths $I_j$ found to achieve the highest throughput) and the corresponding throughput for the joint RCSP approximation. 
Table~\ref{table:m5_steps} also shows that for every value of $k$ the initial code rate $k/I_1$ is above the channel capacity of $0.6851$.  These high initial code rates indicate that feedback is allowing the decoder to capitalize on favorable noise realizations by attempting to decode and terminate early.

\begin{figure}[t]
\centering
    \includegraphics[width=0.5\textwidth]{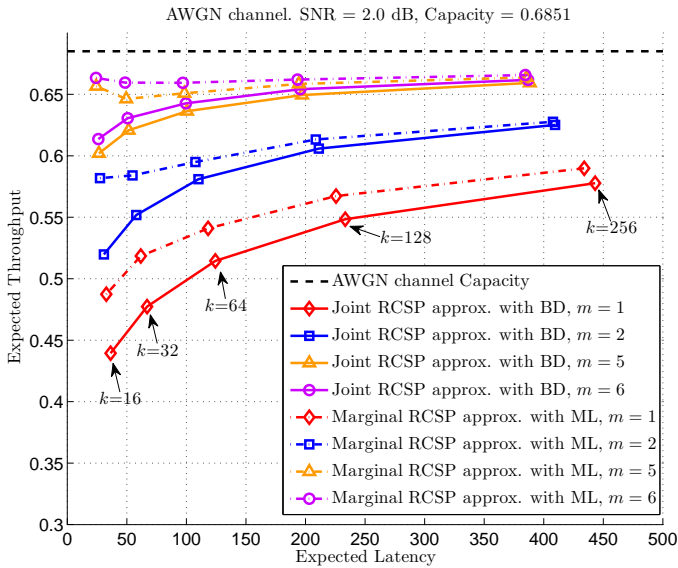}
    \caption{The $R_t$ vs. $\ell$ curves using the joint RCSP approximation with BD decoding and the marginal RCSP approximation with ML decoding for $m = 1, 2, 5, 6$. The transmission lengths $\{I_j\}$ are identified by joint RCSP approximation. }
    \label{fig:latVthroughput}
\end{figure}

Fig.~\ref{fig:latVthroughput} shows the optimized joint RCSP approximation of $R_t$ vs. $\ell$ for $m = 1, 2, 5, 6$ on an AWGN channel with SNR $2$ dB. As $m$ increases, each additional retransmission increases the expected throughput but the amount of that increase diminishes. The points on each curve in Fig.~\ref{fig:latVthroughput} represent values of $k$ ranging from $16$ to $256$ information bits.  Fig.~\ref{fig:latVthroughput} shows, for example, that by allowing up to four retransmissions ($m=5$) with $k=64$, the joint RCSP approximation has a throughput $R_t=0.636$ bits or $93$\% of the original AWGN capacity\footnote{As shown in Sec.~\ref{sec:IR_RCU_RCSP} the finite-latency capacity with NTC is higher than the asymptotic capacity. However, for small $m$ the capacity increase due to NTC is small, and so we include the AWGN capacity as a point of reference.} with an expected latency of $101$ symbols. Similar results are obtained for other SNRs.
 
Fig.~\ref{fig:latVthroughput} also shows the $R_t$ vs. $\ell$ curves for marginal RCSP approximation using ML decoding. The increments in Table~\ref{table:m5_steps} are used in the computation. The curves for ML decoding shows that the effect of NTC starts to manifest at low latencies as $m$ increases. Moreover, the differences between the BD decoding curves and ML decoding curves are negligible for $m>1$ and expected latencies larger than $200$. This observation motivates us to focus on the expected throughput optimization using BD decoding, which simplifies the computation.

\subsection{Chernoff Bounds for RCSP over AWGN Channel}
\label{sec:ChernoffBounds}
Even using the less complex BD decoding, the computation of the joint RCSP approximation based on numerical integration becomes unwieldy for $m > 6$.  In this subsection we study upper and lower bounds based on the Chernoff inequality and compare these bounds with the marginal RCSP approximation of throughput, which is itself a lower bound on the joint RCSP approximation. Similar to previous subsection we assume BD decoding unless otherwise stated. 

Assume w.l.o.g. that the noise has unit variance. Let $r_j$ and $n_j$ be the decoding radius and blocklength for the $j$th decoding attempt. As in earlier sections let $\zeta_{n_j}$ be the marginal error event and $E_{n_i} = \cap_{j=1}^{i}\zeta_{n_j}$ is the joint error event. The main result of applying the Chernoff inequality to bound the joint RCSP approximation is the following theorem.
\begin{theorem}
\label{thm:ChernoffBounds}
Using the joint RCSP approximation with BD decoding for AWGN channel we have for all $1 < i \le m$ that
\begin{align}
\label{eqn:LowerMain}
\P[E_{n_i}] 
&\geq  \max\left\{0, \P[\zeta_{n_i}] - \sum\limits_{j = 1}^{i-1}\P\left[\zeta_{n_i} \cap \zeta_{n_j}^c\right]\right\}\,,
\\
\label{eqn:UpperMain}
\P[E_{n_i}]
&\leq \min\left\{P_1, P_2, 1\right\}\,,
\end{align}
where $P_1 = \P[\zeta_{n_i}]$ and
$P_2 = \P[\zeta_{n_i} \cap \zeta_{n_{i-1}}]$. 
The pairs of joint events $\zeta_{n_m}\cap\zeta_{n_j}, j = 1, \dots, m-1$ can be bounded as follows:
\begin{align}
\label{eqn:TwoJoints1}
&\P[\zeta_{n_j} \cap \zeta_{n_m}] \leq \inf\limits_{0\leq u < 1/2} 
\frac{\P\left[\chi_{n_j}^2 > (1-2u)r_m^2\right]}{e^{u r_m^2}(1-2u)^{n_m/2}}\,,
\\
\label{eqn:TwoJointsComp}
&\P[\zeta_{n_j}^c \cap \zeta_{n_m}] \leq \inf\limits_{0\leq u < 1/2} 
\frac{\P\left[\chi_{n_j}^2 \leq (1-2u)r_m^2\right]}{e^{u r_m^2}(1-2u)^{n_m/2}}\,.
\end{align}
The bounds on pairs of joint events \eqref{eqn:TwoJoints1} and \eqref{eqn:TwoJointsComp} can be extended to joints of more than two events which leads to a slightly tighter upper bound for $\P[E_{n_i}]$. The proof of this extension and the proof of Thm.~\ref{thm:ChernoffBounds} are provided in Appendix~\ref{sec:AppendixChernoff}. 
\end{theorem}
%
%

We observed numerically that the marginal probability $\P[\zeta_{n_m}]$ in \eqref{eqn:UpperMain}, which can be evaluated directly via the tail of a single chi-square random variable, is surprisingly tight for short blocklengths.   The tightness of the marginal was used in \cite{PolyIT11} for \eqref{eqn:AchevVLFT}, where the upper bound on the error probability of each time instance is the marginal. 

Using the marginal $\P[\zeta_{n_i}]$ as an upper bound on $\P[E_{n_i}]$, the lower bound \eqref{eqn:LowerMain} shows that the gap between the joint and marginal probability is upper bounded by applying \eqref{eqn:TwoJointsComp} to $\sum_{j = 1}^{i-1}\P\left[\zeta_{n_i} \cap \zeta_{n_j}^c\right]$. We found numerically that setting $u = 1/2 - n_m/(2r_m^2+2k)$  in \eqref{eqn:TwoJointsComp} where $k = \log_2 M$, gives a tight upper bound on the gap, although this convenient choice of $u$ is not the optimal value. 

\begin{figure}[t]
\centering
	\includegraphics[width=0.5\textwidth]{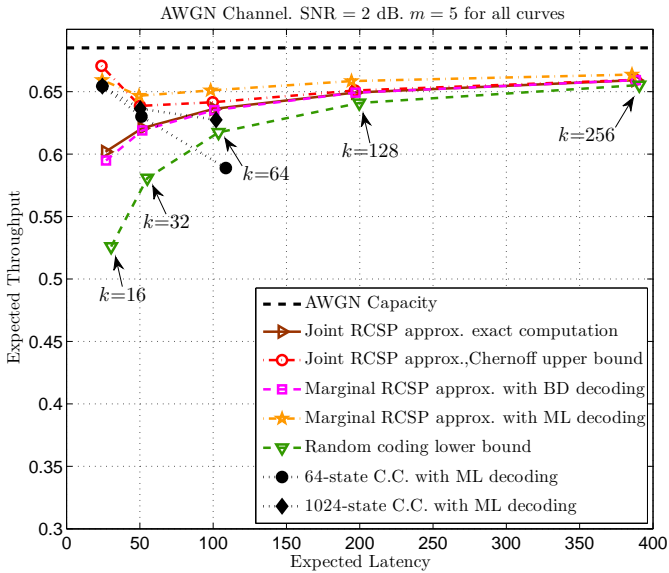}
    \caption{The $R_t$ vs. $\ell$ curves for the joint RCSP approximation with $m = 5$ over the $2$dB AWGN channel. All curves use the optimized increments in Table~\ref{table:m5_steps}.}
    \label{fig:m5compare}
\end{figure}

Fig.~\ref{fig:m5compare} shows the $R_t$ vs. $\ell$ curves for $m=5$ using the optimized step sizes provided in Table~\ref{table:m5_steps} of Sec.~\ref{sec:Optimization} and the values of $k$ are shown in the figure. The channel SNR is $2$ dB and the asymptotic capacity is $0.6851$ bits. The $m=5$ curves shown include exact numerical integration of the joint RCSP approximation, the marginal RCSP approximation using both ML decoding and BD decoding, the upper bound on the joint RCSP approximation using \eqref{eqn:LowerMain} and \eqref{eqn:TwoJointsComp}, and the random-coding lower bound using \eqref{eqn:VLFT_DT}. Evaluating \eqref{eqn:AchevVLFT} would give a slightly better bound than \eqref{eqn:VLFT_DT} for random coding but is very time-consuming to compute.  

The throughput upper bound using \eqref{eqn:LowerMain} and \eqref{eqn:TwoJointsComp} becomes tight for latencies larger than $100$. The lower bound on the joint RCSP approximation using \eqref{eqn:UpperMain} and \eqref{eqn:TwoJoints1} is not shown separately because it turns out to be identical to  the marginal RCSP approximation.  This is because the Chernoff bound of the pairwise joint probabilities are often larger than the marginal probabilities.  

Fig.~\ref{fig:m5compare} also plots the $R_t$ vs. $\ell$ points of the $1024$-state and $64$-state convolutional codes presented in Sec.~\ref{sec:SimSetup}. The simulation results demonstrate that both codes achieve throughputs higher than the random-coding lower bound for  $k = 16$ and $k=32$. The more complex $1024$-state code gives expected throughput higher than random coding even for $k=64$. For $k = 16$ and $k=32$ the simulation points of the $1024$-state code closely approach the marginal RCSP approximation of the expected throughput using ML decoding.

 Note that the lower bound of the expected throughput using random coding is significantly below the RCSP approximations (joint RCSP approximation using BD decoding and marginal RCSP using ML or BD decoding) for low expected latencies.  However, for expected latencies above $350$ symbols the RCSP approximations and the lower bound based on random coding produce very similar expected throughputs.
 
\begin{figure}[t]
\centering
    \includegraphics[width=0.5\textwidth]{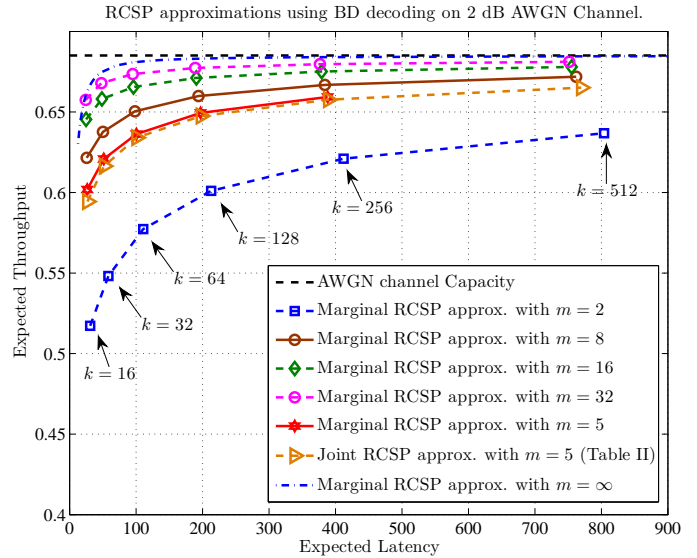}
    \caption{Comparing marginal RCSP approximation with optimized uniform increment $I_j =I$ and various $m$, and joint RCSP approximation with optimized $\{I_j\}_{j = 1}^m$ and $m = 5$.}
	\label{fig:RCSP_k512_m64}
\end{figure}

For random coding, an i.i.d. codebook is drawn using a Gaussian distribution with a zero mean and a variance equal to the power constraint $\eta$.  This type of random codebook generation will sometimes produce a codeword that violates the power constraint.  To address this, the average power should be slightly reduced or  codewords violating the power constraint should be purged, either of which will lead to a slight performance degradation. 

To conclude this subsection, we summarize two relevant observations to motivate the next subsection: (1) When using BD decoding, the difference in expected throughput between the joint RCSP approximation and the marginal RCSP approximation is negligible. (2) The difference between the joint RCSP approximation using BD decoding and marginal RCSP approximation using ML decoding is small for reasonably large expected latencies, e.g., $200$ for $m > 1$. These two observations allow us to focus on efficient optimizations based on marginal RCSP approximation using BD decoding.

\subsection{RCSP with Uniform Increments}
\label{sec:RCSP_FixedK}

For a specified $k$ and for a fixed finite number of transmissions $m$, there are $m$ variables $\{I_j\}_{j = 1}^{m}$ that can be varied to optimize the throughput. The number of possible combinations of $I_j$'s increases rapidly as $m$ increases. Sec.~\ref{sec:OptimizeNonUniformI} addressed the optimization problem for $m \le 6$ by using the joint RCSP approximation.  Motivated by the pattern seen in Table~\ref{table:m5_steps}, this subsection considers the large $m$ case by restricting the transmissions to use uniform increments $I_j = I$ for $j>1$.   This yields a two parameter optimization: the initial blocklength $n_1$ and the increment $I$.  To reflect practical constraints, we restrict the increment $I$ to be an integer.

We also reduce computational burden by replacing the joint RCSP approximation with the marginal RCSP approximation, and we only use BD decoding.  Note that in Sec.~\ref{sec:ChernoffBounds} we saw that the marginal RCSP approximation is operationally identical to lower bound on the joint RCSP approximation using \eqref{eqn:UpperMain} and \eqref{eqn:TwoJoints1} and is a tight lower bound to the joint RCSP approximation.  It is also relatively simple to compute using BD decoding since the relevant probability of error is simply the tail of a chi-square.  

Fig.~\ref{fig:RCSP_k512_m64} presents the optimized performance with uniform increments for various $m$ ranging from $2$ to $\infty$.  In the optimization that produced this figure, the longest possible blocklength was constrained to be less than $\lceil 6k/C \rceil$ where $C$ is the capacity of the AWGN channel. 

\begin{table}[t]
\centering
\caption{Optimized RCSP $n_1$ and $n_m$ for $m=5$ and uniform increments.} 
\begin{tabular}{ c |c c c | c}
  $k$ & $n_1$ & $n_5$ & $I$ & $R_t$ \\
  \hline
  16 &  17 & 47 & 6 &0.5944\\
  \hline
  32 &  39 & 79 & 8  & 0.6164\\
  \hline
   64 & 83 & 143 & 12  & 0.6341\\
  \hline
   128 & 172 & 262 & 18  & 0.6475\\
  \hline
   256 & 353 & 483 & 26  & 0.6576\\
\end{tabular}
\label{table:m5_UnifSteps}
\end{table}

For the $m=5$ case, it is instructive to compare optimized uniform increments with the unconstrained optimal increments of Table \ref{table:m5_steps}.  Table~\ref{table:m5_UnifSteps} shows the numerical results of the uniform-increment optimization for $m=5$ on the $2$ dB AWGN channel.  Comparing the $m=5$ curves in Fig.~\ref{fig:RCSP_k512_m64} and the parameters $n_1$, $n_m$, $I$,  and $R_t$ in Tables \ref{table:m5_steps} and \ref{table:m5_UnifSteps} shows that the constraint of constant increments $I_j = I$ for $j>1$ negligibly reduces expected throughput in this case.

Looking at the uniform-increment curves in Fig. \ref{fig:RCSP_k512_m64}, we observe diminishing returns even for $m$ increasing exponentially. This implies that for a practical system it suffices to consider an $m$ smaller than $16$.

\begin{figure}[t]
\centering
    \includegraphics[width=0.5\textwidth]{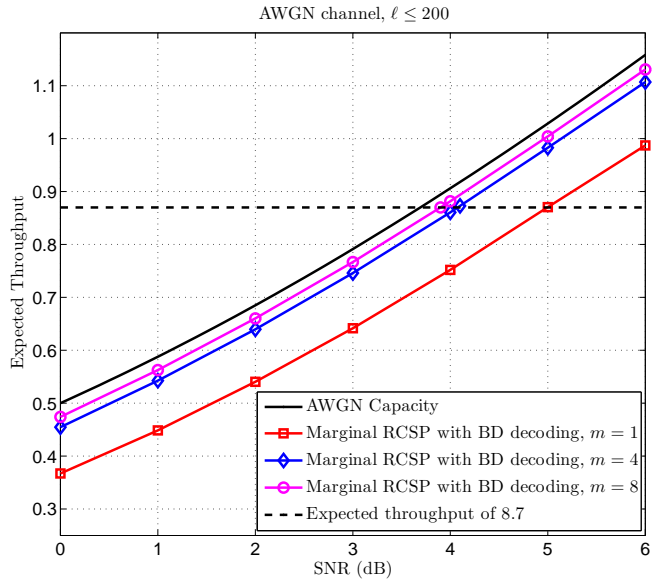}
	 \caption{$R_t$ vs. $\eta$ for $\ell \leq 200$, $\eta = 0, 1, \dots, 5$dB and $m = 1, 4, 8$. }
    \label{fig:VariousSNR_Absolute}
\end{figure}

\subsection{Performance across a range of SNRs}
To allow easy comparison across the various plots above, we have focused attention on the specific case of the 2 dB AWGN channel.  The uniform-increment approach of the previous subsection allows us to efficiently explore performance across a range of SNRs.  For expected latencies constrained to be close to (but not greater than) $200$ symbols, Fig.~\ref{fig:VariousSNR_Absolute} plots the marginal RCSP approximation of $R_t$ vs. $\eta$ for $m = 1, 4, 8$ and $\eta$ ranging from $0$ to $5$ dB. The expected throughput $R_t$ is obtained by finding the largest integer $k$ such that the optimized initial blocklength $n_1$ and the uniform increment $I$ yield expected latency $\ell \leq 200$.  The actual expected latencies ranged only between $197$ and $200$.  We chose the constraint to be $200$ since the difference between ML decoding and BD decoding for the marginal RCSP approximation is small. 

This plot shows the significant benefit of even limited IR as compared to ARQ over a range of SNRs.  For example, at 4 dB the curve for $m = 1$ (ARQ) is $0.155$ bits from the original AWGN capacity, but this gap reduces to $0.046$ bits for $m=4$ and $0.025$ bits for $m = 8$. To see these gaps from an SNR perspective, the horizontal line at expected throughput $8.7$ bits shows that $m=1$ (ARQ) performs within $1.3$ dB of the original AWGN capacity while $m=4$ is within $0.4$ dB and $m=8$ is within $0.2$ dB.  

Recall that the marginal RCSP curves in Fig.~\ref{fig:VariousSNR_Absolute} are for repeated IR-NTC, which generally can have throughputs above capacity because of the extra information communicated by the NTC.  We saw this in Figs. \ref{fig:BSC_Conv_Eg} and \ref{fig:AWGN_RCSP_VLFT_mInfty}.  However, this extra information is quite limited for small values of $m$.  A simple upper bound on the extra information per transmitted symbol to communicate $\tau$ for repeated IR-NTC with initial blocklength $n_1$ is $\log_2(m+1)/n_1$. We examine this upper bound for the case of $4$ dB. For $m=1$, $n_1$ is $192$ symbols and the upper bound is $0.0052$ bits.  For $m=4$ and $m=8$ the values of $n_1$ are $182$ and $178$ respectively, and the upper bounds on the NTC per-symbol extra information are $0.0127$ and $0.0177$ bits respectively.

Thus, the small values of $m$ along with practically reasonable expected latencies of around $200$ symbols considered in Fig.~\ref{fig:VariousSNR_Absolute} cause the extra information provided by NTC to be negligible.   Note that for larger expected latencies, the per-symbol extra information of NTC will become even smaller.

\subsection{Optimizing Increments for Non-Repeating IR-NTC}
\label{sec:outage} 
Repeated IR-NTC has an outage probability of zero because it never stops trying until a message is decoded correctly.  However, this leads to an unbounded maximum latency.  Using a non-repeating IR-NTC scheme, optimization of transmission lengths using the joint RCSP approximation can incorporate a strict constraint on the number of incremental transmissions so that the transmitter gives up after $m$ transmissions.  This optimization can also include a constraint on the outage probability, which is nonzero for non-repeating IR-NTC.

To handle these two new constraints, we fix $m$ and restrict $\P[E_{n_m}]$ to be less than a specified $p_{\text{outage}}$. Without modifying the computations of $\P[E_{n_j}]$, the optimization is adapted to pick the set of lengths that yields the maximum throughput s.t. $\P[E_{n_m}] \leq p_{\text{outage}}$.  When there is a decoding error after the $m$th transmission, the transmitter declares an outage event and proceeds to encode the next $k$ information bits. This scheme is suitable for delay-sensitive communications, in which data packets are not useful to the receiver after a fixed number of transmission attempts. The expected number of channel uses $\ell$ is given by
\begin{equation}
\ell = I_1+\sum\limits_{j=2}^{m} I_j \P[E_{n_{j - 1}}].
\label{eqn:lambda_tau_noRep}
\end{equation}
The expected throughput $R_t$ is again given by $k/\ell$ and the outage probability is $\P[E_{n_m}]$.

\begin{figure}[t]
  \centering
   \includegraphics[width=0.5\textwidth]{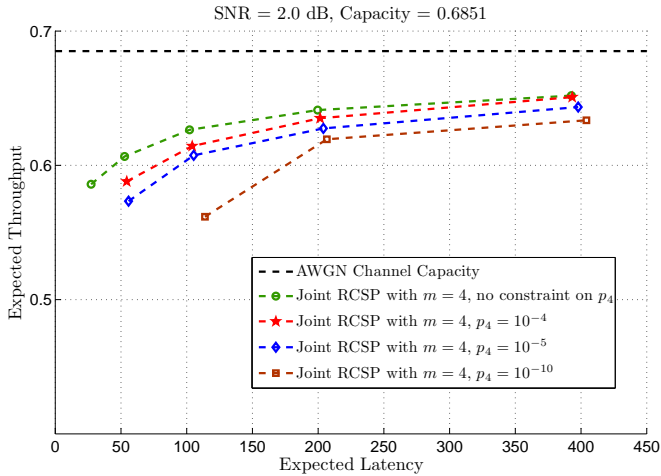}
	\caption{The effect of specifying a constraint on the outage probability $p_4$ on the latency vs. throughput.}
\label{fig:Outage}
\end{figure}
\begin{figure}[h]
\centering
	\includegraphics[width=0.5\textwidth]{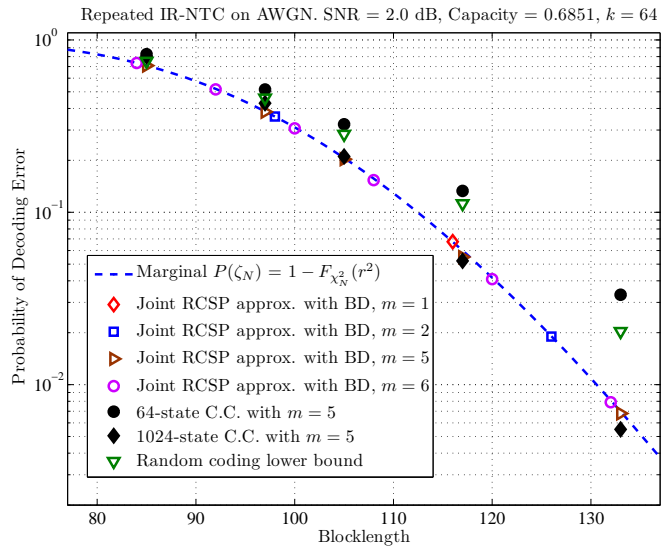}
\caption{A comparison of the decoding error trajectories of joint RCSP approximation, marginal RCSP approximation, simulated ML-decoded convolutional codes and random-coding lower bound for $k=64$.}
\label{fig:PrErrorM5K64}
\end{figure}

Fig.~\ref{fig:Outage} shows how the outage probability constraint affects the $R_t$ vs. $\ell$ curve. The maximum number of transmissions is fixed to be $m = 4$.  The constraint values we considered for the outage probability (error probability of the fourth transmission $\P[E_{n_4}] = p_4$)  are $1 \text{ (no constraint)}, 10^{-4}, 10^{-5}$ and $10^{-10}$. Stricter constraints on outage probability increase the average latency $\ell$. According to the joint RCSP approximation, however, it is exciting to see that the loss in the expected throughput is only $0.022$ bits around latency of $200$ symbols compared to the unconstrained $R_t$ even with the outage constraint $p_4 = 10^{-10}$.


\subsection{Decoding Error Trajectory with BD Decoding}

The optimization of increments in Sec.~\ref{sec:OptimizeNonUniformI} uses the joint RCSP approximation\footnote{Note that for joint RCSP approximation we only use BD decoding.} to find the highest expected throughput $R_t$. The joint RCSP approximation provides a set of joint decoding error probabilities $\P[E_{n_j}], j = 1, \dots, m$, which we call the ``decoding error trajectory''.  If we can find a family of rate-compatible codes that achieves this decoding error trajectory, then we can match the throughput performance suggested by the joint RCSP approximation. For the short-latency regime, e.g. $k = 16$ and $k=32$, one should use the marginal RCSP approximation with ML decoding to study the decoding error trajectory for a better approximation. To demonstrate an example of the decoding error trajectory we focus our attention on the case of $k=64$ and use BD decoding throughout this subsection.

Fig.~\ref{fig:PrErrorM5K64} presents the decoding error trajectories for $k=64$ and $m = 1, 2, 5, 6$ using the joint RCSP approximation. Each trajectory corresponds to a $k=64$ point on the $R_t$ vs. $\ell$ curve in Fig.~\ref{fig:latVthroughput}. For example, the decoding error trajectory for $k=64$ and $m = 2$ consists of the two blue square markers in Fig.~\ref{fig:PrErrorM5K64}  and corresponds to the point on the blue solid curve in Fig.~\ref{fig:latVthroughput} with $k = 64, m = 2$.

Fig.~\ref{fig:PrErrorM5K64} also shows the decoding error trajectories for the random-coding lower bound using \eqref{eqn:VLFT_DT} with $m=5$, as well as the simulations of the two tail-biting convolutional codes presented in Sec.~\ref{sec:SimSetup} with $m=5$.  The dashed line is the decoding error trajectory using the marginal RCSP approximation. The marginal RCSP approximation provides a good estimate that can serve as a performance goal for practical rate-compatible code design across a wide range of blocklengths.

While the $64$-state code is not powerful enough to match the trajectory suggested by the joint RCSP approximation, the $1024$-state code closely follows the trajectory for $m=5$ and therefore has a performance very close to the joint RCSP (c.f. Fig.~\ref{fig:m5compare}). Thus there exist practical codes, at least in some cases, that achieve the idealized performance of RCSP.

Fig.~\ref{fig:PrErrorM4K128} shows how the outage probability constraints affect the decoding error trajectory for $k = 128$ and $m = 4$ using the joint RCSP approximation. Curves are shown with no constraint on the outage probability $p_4$ and for $p_4$ constrained to be less than $10^{-4}, 10^{-5}$ and $10^{-10}$. For ease of comparison, the $x$-axis is labeled with the transmission index rather than blocklength as in Fig.~\ref{fig:PrErrorM5K64}.  

At each index, the blocklengths corresponding to the curves with different constraints are different. For example, at transmission index two, the curve with $p_4 = 10^{-4}$ has blocklength $191$ whereas the curve with $p_4 = 10^{-10}$ has blocklength $211$.  An important observation is that even for relatively low outage probability constraints such as $p_4 = 10^{-10}$, the initial transmission  should still have a relatively high decoding error rate in order to take advantage of instantaneous information densities that may be significantly higher than capacity. 

\begin{figure}[t]
\centering
    \includegraphics[width=0.5\textwidth]{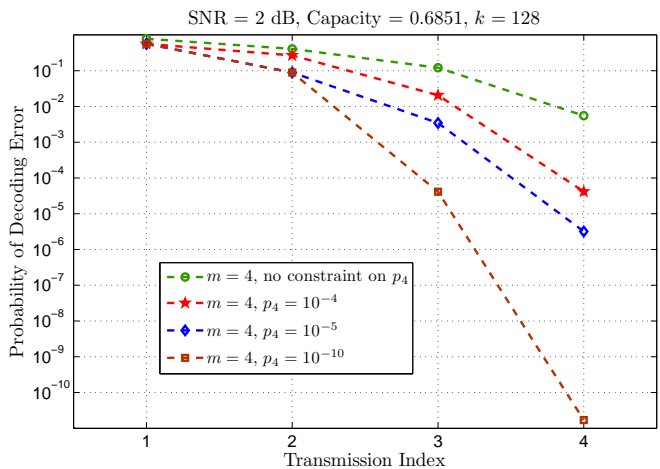}
\caption{The effect of specifying a constraint on the outage probability $\P[E_4] = p_4$ on the decoding error trajectory for $k=128$.}
\label{fig:PrErrorM4K128}
\end{figure}

\section{Concluding Remarks}
\label{sec:Conclusion}

Inspired by the achievability and converse results in \cite{PolyIT11} and practical simulation results in \cite{Chen_2010_ITA}, this paper studies feedback communication systems that use incremental redundancy. We focus on the convenient model of IR-NTC, in which a stream of incremental redundancy concludes when a noiseless confirmation symbol is sent by the transmitter once the receiver has successfully decoded.  

VLFT achievability in \cite{PolyIT11} uses a non-repeating IR-NTC system with an infinite-length mother code and decoding attempted after each received symbol.  The first part of this paper shows that a finite-length mother code (implying repeated IR-NTC to achieve zero-error communication) with decoding  attempted  only at certain specified times can still approach the VLFT achievability curve.  The finite-length constraint introduces only a slight penalty in expected latency as long as the additional length of the mother code beyond the blocklength corresponding to capacity grows logarithmically with the expected latency.  This is a requirement that is easily met by practical systems.  In contrast, the expected latency penalty associated with decoding time limitations is linear in the interval between decoding times.  This forces the intervals to grow sub-linearly in the expected latency for systems to approach capacity.  


The second part of this paper introduces rate-compatible sphere-packing (RCSP) and uses this tool  to analyze and optimize IR-NTC systems for the AWGN channel.  The joint RCSP approximation with BD decoding optimizes the incremental lengths $I_1, \dots, I_m$ for small values of $m$ in a repeated IR-NTC system.  We found that under BD decoding, the marginal RCSP approximation is a tight lower bound of the joint RCSP approximation of expected throughput and simplifies the computation. This simplification allows optimization of the uniform incremental length $I$ for repeated IR-NTC with larger values of $m$.  The marginal RCSP approximation can also be computed for ML decoding, and the difference between ML decoding and BD decoding is significant for short expected latencies. For expected latencies larger than $200$ symbols, however, we observed that the difference between ML and BD decoding becomes small. 

For relatively small values of $m$ and $N$, a repeated IR-NTC system can approach the capacity of the original AWGN channel with expected latencies around 200 symbols.   We applied the marginal RCSP approximation assuming BD decoding across a range of SNRs to an IR-NTC system with $m=8$ and expected latencies at or below 200 symbols.    The results showed throughputs consistently within about 0.2 dB of the performance corresponding to the original AWGN capacity.    The NTC introduces additional information that can generally cause IR-NTC achievable rates to be above capacity.  When $m$ is less than $8$ and the expected latency is above $200$ symbols, however, this increase in throughput is limited to negligible values (less than $0.02$ bits) so that comparisons with the original AWGN channel capacity are reasonable.

For non-repeating IR-NTC, we can use the joint RCSP approximation with BD decoding to optimize the incremental lengths $I_1, \dots, I_m$ under an outage constraint. Numerical result shows that for an expected latency above $200$ symbols, strict outage probability constraints can be met with minimal loss in throughput.

From a practical code design perspective, this paper demonstrates an IR-NTC system for $m=5$ incremental transmissions based on a $1024$-state, randomly punctured, tail-biting convolutional code with optimized transmission increments. At short expected latencies, the resulting IR-NTC system exceeds the random-coding lower bound of \cite{PolyIT11} and closely matches the throughput-latency performance predicted by RCSP for the AWGN channel at low latency.

Rate-compatible codes for IR-NTC systems that match the performance predicted by RCSP remain to be identified for expected latencies between $200$ and $600$ symbols.  This paper demonstrates that the decoding error trajectory based on the marginal RCSP approximation can provide the target error probabilities for designing such rate-compatible codes. Approximations based on both ML decoding and BD decoding can be used. BD decoding is easier to compute and we showed that the difference between ML and BD decoding becomes small for $m>1$ and expected latencies larger than $200$. The design of rate-compatible codes matching the marginal-RCSP decoding error trajectory for expected latencies between $200$ and $600$ symbols is a challenging open problem in channel code design.

\appendices
\section{Proofs for VLFT with Practical Constraints}
\label{sec:AppendixVLFT}

\begin{proof}[Proof of Thm. \ref{thm:FiniteVLFT}]
Consider a random codebook $\mc{C}_N = \{\ms{C}_1, \dots, \ms{C}_M\}$ with $M$ codewords of length-$N$ and codeword symbols independent and identically distributed according to $P_X$. To construct a VLFT code consider the following $(U, f_n, g_n, \tau)$:
The common random variable
\begin{align}
U \in \mc{U} = \overbrace{\mc{X}^N\times\dots\times\mc{X}^N}^{M \text{times}}.
\end{align}
is distributed as:
\begin{align}
U \sim \prod_{j=1}^M P_{X^N}.
\end{align}
A realization of $U$ corresponds to a deterministic codebook  $\{\ms{c}_1, \dots, \ms{c}_M\}$. 
Let $\ms{C}_W(n)$ denote the $n$th symbol of the codeword $\ms{C}_W$, and let $[\ms{C}_j]^n$ denote the first $n$ symbols of the codeword $\ms{C}_j$. The sequence $(f_n, g_n)$ is defined as
\begin{align}
f_n(U, W) &= \ms{C}_W(n)
\\
g_n(U, Y^n) &=  \arg\max_{j = 1, \dots, M} i([\ms{C}_j]^n; Y^n),
\end{align}
and the stopping time $\tau$ is defined as:
\begin{align}
\tau &= \inf\{n: g_n(U,Y^n) = W\}\wedge N\,.
\end{align}
The $n$th marginal error event $\zeta_n$ is given as:
\begin{align}
\zeta_n = \left\{\bigcup_{j \ne W}i(\ms{C}_j^n; Y^n) > i(\ms{C}_W^n; Y^n)\right\}\,.
\end{align}
Following \eqref{eqn:SumsOfPtau}-\eqref{eqn:MarginalZeta} we have 
\begin{align}
\E[\tau] &= \sum_{n = 0}^{N-1} \P[\tau > n]
\\
&\leq \sum_{n = 0}^{N-1}  \P[\zeta_n].
\end{align}
As in \cite[(151)-(153)]{PolyIT11}, the union bound $\P(\zeta_n) \le \xi_n$ provides an upper bound on \eqref{eqn:MarginalZeta} as follows:
\begin{equation}
\E[\tau]\leq \sum_{n = 0}^{N-1} \xi_n\,,
\end{equation}
where $\xi_n$ is given in \eqref{eqn:Xi_n}.

With a similar bounding technique, the error probability can be upper bounded as:
\begin{align}
\P[g_\tau(U, Y^\tau)\ne W]
&= \P[g_N(U, Y^N) \ne W, \tau = N]
\\
&= \P\left[\bigcap_{j = 1}^N \zeta_j\right]
\\
&\leq \P[\zeta_N]
\\
&\leq  \xi_N\,.
\end{align}
In other words, the error probability is upper bounded by the error probability of the underlying finite-length code $\mc{C}_N$.
\end{proof}

\begin{proof}[Proof of Thm. \ref{thm:FiniteFV}]
The proof follows from random coding and the following modification of the triplet $(f_n, g_n, \tau)$ of Thm. \ref{thm:FiniteVLFT}: For $k = 1, 2, \dots $ let $(f'_n, g'_n)$ be defined as:
\begin{align}
\nonumber
f'_n(U,W) &= \begin{cases} 
   f_n(U, W) &\text{ if } n \leq N  \\
   f_{n-kN}(U, W)    &\text{ if } kN < n \leq (k+1)N
 \end{cases}
\\
\nonumber
g'_n(U,Y^n) &= \begin{cases} 
   g_n(U,Y^n)  &\text{ if } n \leq N  \\
   g_{n-kN}(U,Y_{kN+1}^n)   &\text{ if }  kN < n \leq (k+1)N
\end{cases}
\end{align}
Let the new stopping time $\tau'$ be defined as:
\begin{align}
\tau' &= \inf\{n: g'_n(U,Y^n) = W\}\,.
\end{align}
The error probability is zero because the definition of the stopping time $\tau'$ ensures that decoding stops only when the decision is correct. As mentioned above, the new encoder/decoder sequence $(f'_n, g'_n)$ is simply an extension of the VLFT code in Thm.~\ref{thm:FiniteVLFT} by performing an ARQ-like repetition. The expectation of $\tau'$ is thus given as:
\begin{align}
\E[\tau'] &=\sum_{n = 0}^{N-1} \P\left[\bigcap_{j = 1}^{n}\zeta_j\right] + \P\left[\bigcap_{j = 1}^{N}\zeta_j\right]\E[\tau']
\\
&\leq \sum_{n = 0}^{N-1} \P[\zeta_n] + \P[\zeta_N]\E[\tau'] \,,
\end{align}
which implies that:
\begin{equation}
 \E[\tau'] \leq (1-\P[\zeta_N])^{-1}\sum_{n = 0}^{N-1} \P[\zeta_n] \,.
\end{equation}
Applying the RCU bound to replace each $\P[\zeta_n]$ with $\xi_n$ completes the proof.  
\end{proof}

\begin{proof}[Proof of Thm.~\ref{thm:VLFTExpandFinite}]
We define a pair of random walks to simplify the proofs:
\begin{align}
\label{eqn:RW1}
S_n &\triangleq i(X^n;Y^n)
\\
\label{eqn:RW2}
\bar{S}_n &\triangleq i(\bar{X}^n;Y^n)\,.
\end{align}
Referring to \eqref{eqn:InformationDensity}, note that for any measurable function $f$ we have the property:
\begin{align}
\label{eqn:tilting}
\E[f(\bar{X}^n,Y^n)] = \E[f(X^n,Y^n)\exp\{-S_n\}].
\end{align}
Letting $P_{X}$ to be a capacity-achieving input distribution, observe that $S_n$ and $\bar{S}_n$ are sums of i.i.d. r.v.s with positive and negative means:
\begin{align}
\E[i(X;Y)] &= C\\
\E[i(\bar{X};Y)] &= -L \, ,
\end{align}
where $C$ is the channel capacity and $L$ is the lautum information \cite{Palomar08}.  The sequence $\{S_n - nC\}_n$ is a bounded martingale based on our assumption that the information density of each symbol is essentially bounded.  Hence, by Doob's optional stopping theorem we have for a stopping time $\tau$:
\begin{align}
\label{eqn:DoobOpt}
\E[S_\tau] = C\E[\tau]\,.
\end{align}
Properties \eqref{eqn:tilting}-\eqref{eqn:DoobOpt} are used in the rest of this appendix. 

Using the definitions of \eqref{eqn:RW1} and \eqref{eqn:RW2} in \eqref{eqn:Xi_n} produces \eqref{eqn:XiWithS1}.  Weakening the RCU bound using \eqref{eqn:VLFT_DT} and replacing $M-1$ with $M$ in  \eqref{eqn:VLFT_DT} for simplicity produces \eqref{eqn:XiWithS2}:
\begin{align}
\label{eqn:XiWithS1}
\xi_n&=\E\left[\min\left\{(1,(M-1)\P[\bar{S}_n\geq S_n|X^nY^n]\right\}\right] 
\\
\label{eqn:XiWithS2}
&\leq \E\left[\exp\{-[S_n - \log M]^+\}\right]\,.
\end{align} 
Applying \eqref{eqn:XiWithS2} to Thm. \ref{thm:FiniteFV} yields the following:
\begin{align}
\label{eqn:XiWithS3}
\ell \leq \frac{1}{(1-\xi_N)}\sum_{n = 0}^{N-1}\E\left[\exp\left\{-[S_n - \log M]^+\right\}\right]\, .
\end{align}
Consider an auxiliary stopping time $\tilde{\tau}$ w.r.t. the filtration $\mc{F}_n = \sigma\{X^n, \bar{X}^n, Y^n\}$:
\begin{align}
\tilde{\tau} = \inf\{n \geq 0: S_n \geq \log M\}\wedge N\,.
\end{align}

For a specified set $E$, use $\E[X; E]$ to denote $\E[X 1_E]$ where $1_E$ is the indicator function of the set $E$. We now turn our attention to computing the summation in \eqref{eqn:XiWithS3}. Letting $E$ be the set $\{\tilde{\tau}< N \}$,  we have the following:
\begin{align}
\nonumber
	&\sum_{n = 0}^{N-1}\E\left[\exp\{-[S_n - \log M]^+\}\right] 
\\
\nonumber
	&= \E\left[\tilde{\tau}-1 + \sum_{k = 0}^{N-1-\tilde{\tau}}\exp\left\{-[S_{\tilde{\tau}+k} - \log M ]^+\right\};E\right]
\\
\label{eqn:summation}
	&\quad + N\P[E^c]\,.
\end{align}
On $E$ we have $i(X^{\tilde{\tau}};Y^{\tilde{\tau}}) \geq \log M$ and hence:
\begin{align}
\label{eqn:sk1}
\left[S_{\tilde{\tau}+k} - \log M\right]^+ 
&= \left[S_{\tilde{\tau}+k} - S_{\tilde{\tau}} + S_{\tilde{\tau}} - \log M\right]^+
\\
\label{eqn:sk2}
&\geq  \left[S_{\tilde{\tau}+k} -  S_{\tilde{\tau}}\right]^+
\\
\label{eqn:sk3}
& {\buildrel d \over =} \left[S_{k}\right]^+ 
\end{align}
where the last equality is equality in distribution and is true almost surely by the strong Markov property of random walks. Applying \eqref{eqn:sk1}-\eqref{eqn:sk3} to \eqref{eqn:summation} yields:
\begin{align}
\nonumber
&\sum_{n = 0}^{N-1}\E\left[\exp\{-[S_n - \log M]^+\}\right] 
\\
\nonumber
&\leq 
\E\left[\tilde{\tau}-1 + \sum_{k = 0}^{N-1-\tilde{\tau}}\exp\{-\left[S_k\right]^+ \};E\right] 
\\
&\quad+ N\P[ E^c ]\,.
\end{align}
Using \eqref{eqn:tilting} we have the following:
\begin{align}
\E\left[\exp\{-\left[S_k\right]^+ \} \right]
&= \P\left[\bar{S}_k > 0\right] + \P\left[S_k\leq 0\right].
\end{align}

$S_k$ and $\bar{S}_k$ are sums of i.i.d. r.v.s with positive and negative means, respectively. 
Thus by the Chernoff inequality, both terms decay exponentially in $k$, yielding  
\begin{align}
\P\left[\bar{S}_k > 0\right] + \P\left[S_k\leq 0\right] \leq a_1e^{-ka_2} \, ,
\end{align}
for some positive constants $a_1$ and $a_2$.

 Thus there is a constant $a_3 > 0$ such that:
\begin{align}
\label{eqn:a3start}
\sum_{k = 0}^{N-1-\tilde{\tau}}\E\left[\exp\{-[S_k]^+\} \right]
&\leq \sum_{k = 0}^{N-1}\E\left[\exp\{-[S_k]^+\} \right] \\
&\leq \sum_{k = 0}^{N-1} a_1e^{-ka_2} 
\\\label{eqn:a3}
&=\frac{a_1 e^{-a_2}(1-e^{-(N-1)a_2})}{1-e^{-a_2}}\\
\label{eqn:a3stop}
&=a_3\, .
\end{align}
We assume that $S_n$ has bounded jumps, and hence on the set $E$ there is a constant $a_4$ such that
\begin{align}
S_{\tilde{\tau}} - \log M \leq a_4 C \,. 
\end{align}
Therefore from \eqref{eqn:DoobOpt} we have that  on the set $E$:
\begin{align}
\label{eqn:a4}
\E[\tilde{\tau}] \leq \frac{\log M }{C} + a_4 \,.
\end{align}
We are now ready to provide a bound on \eqref{eqn:XiWithS3}.  Letting $a_5=a_3+a_4$ and applying \eqref{eqn:a3start}-\eqref{eqn:a3stop} and \eqref{eqn:a4} to \eqref{eqn:summation} we have:
\begin{align}
\label{eqn:EllBoudnAtN}
\ell \leq (1-\P[\zeta_N])^{-1}\left(\frac{\log M}{C} + a_5 + N\P[E^c] \right) \, .
\end{align} 
For a fixed $M$ with random coding, there is a constant $\Delta > 0$ such that with $C_\Delta = C - \Delta$ and $N = \log M / C_\Delta$  we have the following bound on error probability:
\begin{align}
\label{eqn:ErrorAtN}
\P[\zeta_N]\leq b_2\exp(-N b_3)\, ,
\end{align}
for some constants $b_2>0$ and $b_3 >0$.

\setcounter{equation}{112}
\begin{figure*}[t]
\normalsize
\begin{align}
\label{eqn:6.1}
(1-\xi_{N}) \ell
&\leq n_1 + I\sum_{j = 1}^{m-1}\E[\exp\{-[S_{n_j} - \log M]^+\}]
\\
&=  
\E\left[  n_1+(\tilde{j}-1)I + I\sum_{k = 0}^{m-1-\tilde{j}}\exp\left\{-\left[S_{n_{k+\tilde{j}}} - \log M\right]^+\right\};E\right] + \E[n_1 + (m-1)I;E^c] 
\\
\label{eqn:6.3}
& \leq  \E[\tilde{\tau} ;E]
+ I\E\left[\sum_{k = 0}^{m-1-\tilde{j}}\exp\left\{-\left[S_{n_{k+\tilde{j}}} - \log M\right]^+\right\};E\right] + N\P[E^c]
\\
& \leq  \E[\tilde{\tau} ;E]
+ I\E\left[\sum_{k = 0}^{m-1}\exp\left\{-\left[S_{n_{k}}\right]^+\right\};E\right] + N\P[E^c]
\\
\label{eqn:6.4}
&\leq \frac{\log M}{C} + N\P[E^c]] + O(I).
\end{align}
\hrulefill
\end{figure*}
\setcounter{equation}{90}

Recalling that $S_n$ is a sum of i.i.d. r.v.'s with mean $C$, let  $N = \log M / C_\Delta$ we have by the Chernoff inequality that: 
\begin{align}
\P[E^c] & = \P[S_N < \log M]
\\
&\leq b_0 \exp\left\{-b_1 N\left(C - \frac{\log M }{N}\right)\right\} 
\\ \label{eqn:ChernoffAtN}
& = b_0 \exp\left\{-b_1\frac{\log M}{C_\Delta}\Delta \right\}\,.
\end{align}
Combining \eqref{eqn:EllBoudnAtN} and \eqref{eqn:ChernoffAtN}  we have the following for $\ell$:
\begin{align}
\ell \leq \frac{ \frac{\log M}{C} 
+a_5 + \frac{\log M}{C_\Delta} \frac{ b_0}{ M^{b_1\Delta/C_\Delta}}}{1-\P[\zeta_N]}\,.
\end{align}
Now applying \eqref{eqn:ErrorAtN}, notice that we are only interested in the first two terms of the expansion $(1-\P[\zeta_N])^{-1} = 1 + \P[\zeta_N] + \P[\zeta_N]^2 + \dots$ on $[0,1)$. Thus
\begin{align}
\ell &\leq \frac{\log M}{C} + \frac{ b_0\log M}{C_\Delta M^{b_1\Delta/C_\Delta}}+ \frac{ b_2\log M}{C(M^{b_3/C_\Delta})} + a_6
\end{align}
for some $a_6> 0$. Hence for $M$ large enough we have \eqref{eqn:ellExpansion1}.
\end{proof}

\begin{proof}[Proof of Cor.~\ref{cor:ScalingForEll}]
We first choose $N$ to scale with $\ell$ with a factor $\delta$ to be chosen later:
\begin{align}
N = (1+\delta)\ell\,.
\end{align}
Then by the converse (Thm. \ref{thm:poly11}) we have:
\begin{align}
\frac{\log M}{N} &\leq C + \frac{\log(\ell+1)+\log e - \delta\ell C}{(1+\delta)\ell}
\\
&\leq C-\delta' \, .
\end{align}
The term $\delta'$ on the right is positive by setting:
\begin{align}
\label{eq:mindelta}
\delta > \frac{\log(\ell+1)+\log e}{\ell C}\,.
\end{align}
Again by the Chernoff inequality we have:
\begin{align}
\P[\tau' \geq N] &= \P[S_N  < \log M  ]
\\
&\leq \P[S_N - NC < -N\delta' ]
\\
&\leq b_0'\exp\left\{-\ell(1+ \delta)b_1'\delta' \right\}\, .
\end{align}
Since $\delta$ is chosen such that $\log M / N$ is less than capacity, we also have \eqref{eqn:ErrorAtN}. By reordering \eqref{eqn:EllBoudnAtN} we have for some $b_2', b_3' > 0$ such that:
\begin{align}
\frac{\log M}{C} 
\geq \ell\left[1-b_2e^{-\ell(1+\delta)b_3} - (1+\delta)b_0'e^{-b_2'\ell}\right] - b_3' \,,
\end{align}
which implies $\log M_t^*(\ell, N) \geq \ell C - O(1)$ for $N = (1+\delta)\ell$ and large enough $\ell$.
\end{proof}
\begin{proof}[Proof of Thm.~\ref{thm:VLFT_TimeLimit1}]
Consider the same random-coding scheme as in Thm. \ref{thm:FiniteFV} with encoders $f'_n(U,W)$ and decoders $g'_n(U,Y^n)$.  The auxiliary stopping time $\tilde{\tau}$ of Thm. \ref{thm:VLFTExpandFinite} is altered to reflect the limitation on decoding times as follows:  
\begin{equation}
\tilde{\tau} = n_1 + (\tilde{j}-1)I,
\end{equation}
where $\tilde{j}$ is also a stopping time given as:
\begin{align}
\tilde{j}= \inf\{j > 0: S_{n_j} = i(X^{n_j}; Y^{n_j}) \geq \log M\}\,.
\end{align}
The rest is similar to the proof of Thm. \ref{thm:VLFTExpandFinite}, but without the complication of a finite $N$:
\begin{align}
\ell &\leq n_1 + I\sum_{j = 1}^{\infty}\P[\zeta_{n_j}]
\\
&\leq n_1 + I\sum_{j = 1}^{\infty}\E\left[\exp\left\{-\left[S_{n_j} - \log M\right]^+\right\}\right]
\\
\nonumber
&\leq n_1 + I\E[\tilde{j} - 1]
\\ 
&\quad+ I\sum_{k = 0}^{\infty}\E\left[\exp\left\{-\left[S_{n_{\tilde{j} + k}} - \log M\right]^+\right\}\right]
\\
&\leq \E[\tilde{\tau}] 
+ I\sum_{k = 0}^{\infty}\E\left[\exp\left\{-[S_{n_k}]^+\right\}\right]
\\
\label{eqn:UnifIncBd1}
&\leq \E[{\tilde{\tau} }] + I a_3
\\
\label{eqn:UnifIncBd2}
&\leq \frac{\log M}{C} + I a_4 + I a_3\,,
\end{align}
for some $a_3>0$ and some $a_4>0$, where \eqref{eqn:UnifIncBd1} follows by applying the Chernoff inequality and \eqref{eqn:UnifIncBd2} is a consequence of the jumps $S_{n_k}$ being bounded by $I a_4 C$ for some $a_4 > 0$. 
Reordering the equations gives the result. 
\end{proof}

\begin{proof}[Proof of Lem.~\ref{lem:MainAsympResult1}]
Consider the same random-coding scheme as in Thm. \ref{thm:VLFT_TimeLimit1} with encoders $f'_n(U,W)$ and decoders $g'_n(U,Y^n)$,
an initial blocklength $n_1$, and a uniform increment $I$.  However, the number of increments is now limited to a finite integer $m$. The finite block-length is then given by $N = n_m$ where $n_j = n_1 + (j-1)I$. Similar to Thm. \ref{thm:VLFTExpandFinite}, define the auxiliary stopping time as:
\begin{align}
\tilde{j} = \inf\{j > 0: S_{n_j} \geq \log M\} \wedge m \,.
\end{align}
Similar to Thm.~\ref{thm:VLFTExpandFinite} we have \eqref{eqn:6.1} to \eqref{eqn:6.4}, shown at the top of the page.   In \eqref{eqn:6.1} to \eqref{eqn:6.4}, $E$ is the set $\{\tilde{j} < m\}$.

Let the scaling of $m$ be 
\setcounter{equation}{117}
\begin{align}
m = \left\lceil{\left(\frac{\log M}{IC_\Delta} - 
\frac{n_1}{I}\right) + 1}\right\rceil
\end{align}
which yields a similar choice of $N$ as in the proof of Thm.~\ref{thm:VLFTExpandFinite}:
\begin{align}
\label{eqn:deltaforM}
N = n_1 + (m-1)I \geq \frac{\log M }{C_\Delta}\,.
\end{align}
The rest of the proof follows as in the proof of Thm. \ref{thm:VLFTExpandFinite}.
\end{proof}

\section{Chernoff Bounds for RCSP}
\label{sec:AppendixChernoff}
This appendix gives the proofs of the upper and lower bounds provided in Sec.~\ref{sec:ChernoffBounds}. Let $f_{\chi_{n}^2} = dF_{\chi_n^2}$ be the density function of a chi-square distribution with $n$ degrees of freedom. The following lemma gives upper and lower bounds on the joint of a pair or error events. 

\begin{lemma}
\label{lem:ChernoffTwoTrans}
For an AWGN channel, let $\{\zeta_{n_j}\}_{j = 1}^m$ be the error events of RCSP with bounded-distance decoding. We have the following upper and lower bounds on the probability of a pair of joint error events $\zeta_{n_j}$ and $\zeta_{n_i}$ where $j < i \leq m$:
\begin{align}
\label{eqn:TwoJoints1App}
&\P[\zeta_{n_j} \cap \zeta_{n_i}] \leq \inf\limits_{0\leq u < 1/2} 
\frac{\P\left[\chi_{n_j}^2 > (1-2u)r_i^2\right]}{e^{u r_i^2}(1-2u)^{n_i/2}},
\\
\label{eqn:TwoJointsCompApp}
&\P[\zeta_{n_j}^c \cap \zeta_{n_i}] \leq \inf\limits_{0\leq u < 1/2} 
\frac{\P\left[\chi_{n_j}^2 \leq (1-2u)r_i^2\right]}{e^{u r_i^2}(1-2u)^{n_i/2}},
\\
\label{eqn:TwoJoints2App}
&\P[\zeta_{n_j} \cap \zeta_{n_i}] \geq \max\left(\P[\zeta_{n_j}]-w_1, \P[\zeta_{n_i}]-w_2 \right)\,,
\end{align}
where $w_1$ and $w_2$ are given as:
\begin{align}
w_1 &= \inf\limits_{v\geq 0}\frac{e^{v r_i^2}\P\left[\chi_{n_j}^2 > (1+2v)r_j^2\right]}{(1+2v)^{n_i/2}},
\\
\label{eqn:w2}
w_2 &= \inf\limits_{0\leq v\leq 1/2}\frac{e^{-v r_i^2}\P\left[\chi_{n_j}^2 \leq (1-2v)r_j^2\right]}{(1-2v)^{n_i/2}}.
\end{align}
\begin{proof}
It suffices to show the bounds for $\zeta_{n_1}\cap\zeta_{n_2}$.
Let the set $A_i^j(r_k^2)$ be defined as:
\begin{align}
A_{i}^j(r_k^2) = \left\{z_{i+1}^j:\|z_{i+1}^j\|^2 > r_k^2\right\}.
\end{align} 
Applying the Chernoff inequality to \eqref{eqn:Pzeta_i} for $m = 2$, we have

\begin{align}
&\P[\zeta_{n_1}\cap\zeta_{n_2}]\leq  \int_{r_1^2}^{\infty} \frac{\E e^{u\chi_{I_2}^2}f_{\chi_{I_1}^2}(t_1)}{e^{u(r_2^2 - t_1)}}  dt_1
\\ 
&= \int_{r_1^2}^{\infty} (1-2u)^{-I_2/2}e^{-u(r_2^2 - t_1)} f_{\chi_{I_1}^2}(t_1) dt_1
\\ \label{eqn:ChangeVar}
&=(1-2u)^{-I_2/2}e^{-u r_2^2}
\int_{A_0^{I_1}(r_1^2)}\frac{e^{-\frac{(1-2u)}{2}\sum\limits_{i = 1}^{I_1}z_i^2}}{(2\pi)^{I_1/2}} dz_1^{I_1} 
\\ \label{eqn:ChangeVar2}
&=\frac{\int_{A_0^{I_1}\left((1-2u)r_1^2\right)}\frac{e^{-\sum\limits_{i = 1}^{I_1}z_i'^2/2}}{(2\pi)^{I_1/2}} dz_1'^{I_1}}
{(1-2u)^{I_2/2}(1-2u)^{I_1/2}e^{u r_2^2}}
\\ \label{eqn:UpperBound}
&=\frac{e^{-u r_2^2}\P\left[\chi_{I_1}^2>(1-2u)r_1^2\right]}{(1-2u)^{N_2/2}}\,,
\end{align}
where \eqref{eqn:ChangeVar2} follows from a change of variable $z_i' = (1+2u)^{1/2}z_i$. Taking the infimum over $u < 1/2$ gives the result. 

We provide a sketch of the proof for \eqref{eqn:TwoJointsComp} and the lower bound follows from \eqref{eqn:TwoJointsComp}. Observe that $\P[\zeta_{n_1} \cap \zeta_{n_2}] = \P[\zeta_{n_1}] - \P[\zeta_{n_1} \cap \zeta_{n_2}^c] = \P[\zeta_{n_2}] - \P[\zeta_{n_1} \cap \zeta_{n_2}^c]$. Let $w_1 = \P[\zeta_{n_1} \cap \zeta_{n_2}^c]$, $w_2 = \P[\zeta_{n_1}^c \cap \zeta_{n_2}]$ and finding the upper bounds of them yield the lower bound. The upper bound on $w_2$ and also \eqref{eqn:TwoJointsComp} follows from the above derivation by changing the integration interval from $(r_1, \infty)$ to $[0 , r_1]$. For the upper bound on $w_1$, apply the Chernoff inequality with the form $\P[X\leq r] \leq \E[ e^{-vX}]e^{vr}$. 
Taking the infimum over $v\geq 0$ gives the result.
\end{proof}
\end{lemma}

The following theorem extends the upper bound of the above lemma for the $m$-transmission joint error probability given in \eqref{eqn:Pzeta_i}, which depends on the initial blocklength $n_1 = I_1$ and the $m-1$ step sizes $I_2, \dots, I_m$:
\begin{theorem}
\label{thm:GeneralUpper}
Let $u_i<1/2, i = 1, 2, \dots, m$ be the parameters for each use of Chernoff bound in the integral. Define $h_i, g_i(u_1^m)$ by the following recursion:
\begin{align*}
h_1 &= u_1
\\
g_1 &= e^{-u_1r_m^2}(1-2u_1)^{\frac{-I_{m}}{2}},
\\
h_i &= h_{i-1}+u_i(1-2h_{i-1}),
\\
g_i &= g_{i-1}e^{-u_i(1-2h_{i-1})r_{m-i+1}^2}\left(1-2h_{i-1}\right)^{\frac{-I_{m-i+1}}{2}}.
\end{align*}
Note the property that $1-2h_{i} = \prod_{j \leq i} (1-2u_j)$. We have

\begin{align}
&\P\left[E_{n_m}\right] 
\leq \inf\limits_{u_1^m} \frac{g_{m-1}(u_1^m)\P\left[\chi_{I_1}^2 > (1-2h_{m-1})r_1^2\right] }{(1-2h_{m-1})^{I_1/2}}.
\end{align}
\end{theorem}
\begin{proof}
The proof follows from changing the variables iteratively similar to the proof of Lem.~\ref{lem:ChernoffTwoTrans}.
\end{proof}

Several versions of lower bounds of the joint error probability with $m$ transmissions can be obtained by different expansions of the joint events, and the recursion formulas follow closely to those in Thm.~\ref{thm:GeneralUpper} and Lem~\ref{lem:ChernoffTwoTrans}. The following corollary gives an example of one specific expansion that yields a lower bound in a recursive fashion and the other formulas are omitted. 
\begin{corollary}
\label{cor:GeneralLower}
With the same recursion as in Thm.~\ref{thm:GeneralUpper},  the lower bound is given as:
\begin{equation}
	\label{eqn:GeneralLowerBound}
	\P\left[E_{n_m} \right] \geq 	\max\left\{0, p\right\}\,,
\end{equation}
where $p$ is given as
\begin{equation}
\P\left[\bigcap\limits_{j = 2}^m \zeta_{n_j} \right] - \inf\limits_{u_1^m} \frac{g_{m-1}(u_1^m)\P\left[\chi_{I_1}^2 \leq (1-2h_{m-1})r_1^2\right] }{(1-2h_{m-1})^{I_1/2}}.
\end{equation}
\end{corollary}
\begin{proof}
Expand $E_m$ as 
\begin{equation}
\bigcap_{1\leq j \leq m } \zeta_{n_j}  = \bigcap_{2\leq j \leq m } \zeta_{n_j} \setminus \left(\bigcap_{2\leq j \leq m } \zeta_{n_j} \cap \zeta_{n_1}^c \right)
\end{equation}
and the proof follows from applying Thm.~\ref{thm:GeneralUpper} except for the last event $\zeta_{n_1}$, which gives $\P\left[\chi_{I_1}^2 \leq (1-2h_{m-1})r_1^2\right]$ instead of $\P\left[\chi_{I_1}^2 \geq (1-2h_{m-1})r_1^2\right] $.
\end{proof}

Applying Thm.~\ref{thm:GeneralUpper} for the case of $m = 2$ gives a proof of Thm.~\ref{thm:ChernoffBounds} as shown in the following.

\begin{proof}[Proof of Thm.~\ref{thm:ChernoffBounds}]
We first show that both \eqref{eqn:LowerMain} and \eqref{eqn:UpperMain} follow immediately from properties of probability. For any two sets $A, B$ we can write a disjoint union of A as
\begin{equation}
(A\cap B) \cup (A\cap B^c) = A\,.
\end{equation}
Letting $A = \zeta_{n_m}$ and $B = E_{n_m} \setminus \zeta_{n_m} = E_{n_{m-1}}$, we can rewrite the expression of $\P[E_{n_m}]$ as 
\begin{equation}
\P[E_{n_m}] +  \P[\zeta_{n_m}\cap E_{n_{m-1}}^c]= \P[\zeta_{n_m}]\,.
\end{equation}
We therefore have the following:
\begin{align}
 &\P\left[E_{n_m}\right] = \P[\zeta_{n_m}] - \P\left[\zeta_{n_m}\cap \bigcup_{i= 1}^{m-1}\zeta_{n_i}^c\right]
 \\
 \label{eqn:LowerEg}
 &\geq  \max\left[0, \P[\zeta_{n_m}] - \sum\limits_{j = 1}^{m-1}\P\left\{\zeta_{n_m} \cap \zeta_{n_j}^c\right]\right\},
\end{align}
where the last inequality can be seen as the union bound on the second term of the first equality. 

To show the upper bound, a straightforward probability upper bound of $E_{n_m}$ gives
\begin{align}
\label{eqn:UpperSimple}
\P\left[E_{n_j}\right] \leq \P[\zeta_{n_j}]\,,
\end{align}
which can be computed with the tail of the chi-square CDF directly. Since $\P\left[E_{n_j}\right] \leq \P[\zeta_{n_{j-1}}\cap \zeta_{n_j}]$, applying Lem.~\ref{lem:ChernoffTwoTrans} for upper bounds on $\P[\zeta_{n_{j-1}}\cap\zeta_{n_j}^c]$ finishes the proof.
\end{proof}

Note that alternatively, we can rewrite the joint error probability to obtain different upper bounds. For example, write $\P[E_{n_j}]$, $j \leq  m$ as 
\begin{align}
 &\P\left[\bigcap_{i=1}^{j}\zeta_{n_i}\right] 
 = \P\left[\bigcap_{i=1}^{j}\zeta_{n_i} \cap \zeta_{n_m}\right] +\P\left[\bigcap_{i=1}^{j}\zeta_{n_i} \cap \zeta_{n_m}^c\right]
 \\
 \label{eqn:UpperEg}
 &\leq  \P[\zeta_m, \zeta_j] + \P[\zeta_j, \zeta_{j-1}] - \P[\zeta_j,\zeta_{j-1},\zeta_m]\,.
\end{align}
Applying Thm.~\ref{thm:GeneralUpper} to the first two terms and Cor.~\ref{cor:GeneralLower} to the last term gives another version of the upper bound. Also note that for rates above capacity, only \eqref{eqn:UpperSimple} is active since the Chernoff bounds give trivial results.

\section{The AWGN Power Constraint for RCSP}
\label{sec:PowerConstraintRCSP}
Recall that for an $n$-dimensional input $X^n$ for the AWGN channel with unit noise variance and SNR $\eta$, the power constraint is given as $\sum_{j = 1}^n X_j ^2 \leq n\eta$. Assuming perfect packing of $M = 2^k$ identical Euclidean balls $D_i$ with radii $r$ into the outer sphere with radius $r_{\text{outer}} = \sqrt{n(1+\eta)}$, the radii for $D_i$ is $r = \frac{\sqrt{n(1+\eta)}}{2^{k/n}}$. The codeword point, which is located at the center of the decoding region $D_i$, is at least distance $r$ from the outer sphere surface assuming perfect packing and within the outer sphere.

Let $S_n\left(r\right) = \{y^n : \sum_{i = 1}^n y_i^2 = r^2\}$ be the sphere surface with radius $r$. The set of $x\in\mb{R}^n$ that is at least distance $r$ away from the sphere surface  with radius $r_{\text{outer}} = \sqrt{n(1+\eta)}$ is given as 
\begin{equation}
H(r) = \left\{x^n: \sum_{i = 1}^n (x_i - y_i)^2 \geq r^2, \forall y\in S_n\left(r_{\text{outer}}\right)\right\}.
\end{equation}
Assuming $R = k/n \leq \frac{1}{2}\log_2(1+\eta)$ we have $2^{2k/n} \leq 1+\eta$ and hence
\begin{align}
r^2 &= \frac{n(1+\eta)}{2^{2k/n}} 
\\
&\geq n\,.
\end{align}
Therefore by setting $R = 1/2\log(1+\eta)$ the set $H(r)$ becomes
\begin{align}
H(r) = \left\{x^n: \sum_{i=1}^n (x_i - y_i)^2 \geq n, y^n\in S_n\left(\sqrt{n(1+\eta)}\right) \right\}.
\end{align}
Let $B_n(r)$ denote a ball with radius $r$. With an additional constraint that $x^n$ must also be in the ball with radius $\sqrt{n(1+\eta)}$, we conclude that the codeword points must be in the ball $B_n\left(\sqrt{n\eta}\right)$.
The maximum energy of a codeword is therefore within the power constraint $\sum_{j = 1}^n X_j ^2 \leq n\eta$ if $k/n \leq \frac{1}{2}\log_2(1+\eta)$, the capacity of the AWGN channel. 

\bibliographystyle{IEEEtran}
\bibliography{IEEEabrv,Feedback_Journal}

\end{document}